\documentclass[12pt]{article}
\usepackage[centertags]{amsmath}
\usepackage{amssymb}
\usepackage{amsthm}
\usepackage{graphicx}
\newtheorem{tw}{Theorem}[section]
\newtheorem{lem}[tw]{Lemma}
\newtheorem{cor}[tw]{Corollary}
\newtheorem{prop}[tw]{Proposition}
\newtheorem{defi}{Definition}[section]
\newtheorem*{rem}{Remark}

\theoremstyle{remark}
\newtheorem{example}{Example}[section]

\def\={\hspace{-3mm}&=&\hspace{-3mm}}

\renewcommand{\baselinestretch}{1,2}

\title{A Laplace transform approach to linear equations with infinitely many 
derivatives and zeta-nonlocal field equations}
\author{ A. Ch\'avez$^1$\footnote{E-mail: alancallayuc@gmail.com}~, 
H. Prado$^2$\footnote{E-mail: humberto.prado@usach.cl}~,
and E.G. Reyes$^3$\footnote{E-mail: enrique.reyes@usach.cl ;
e\_g\_reyes@yahoo.ca} \\
\small{$^{1}$Departamento de Matem\'aticas, Facultad de Ciencias,} \\
\small{Universidad de Chile,}\\
\small{ Casilla 653, Santiago, Chile}\\
\smallskip
\small{$^{2,3}$ Departamento de Matem\'atica y Ciencia de la
Computaci\'on,} \\
\small{ Universidad de Santiago de Chile }\\
\small{Casilla 307 Correo 2, Santiago, Chile }}

\begin{document}

\maketitle

\begin{abstract}
We study existence, uniqueness and regularity of solutions for linear equations in 
infinitely many derivatives.
We develop a natural framework based on Laplace transform as a correspondence 
between appropriate $L^p$ and
Hardy spaces: this point of view allows us to interpret rigorously operators of the 
form $f(\partial_t)$ where 
$f$ is an analytic function such as (the analytic continuation of) the Riemann zeta 
function. We find the most general solution to the equation
\begin{equation*}
f(\partial_t) \phi = J(t) \; , \; \; \; t \geq 0 \; ,
\end{equation*}
in a convenient class of functions, we define and solve its corresponding 
initial value problem, and we state conditions under which the
solution is of class $C^k,\, k \geq 0$. More specifically, we prove that
if some a priori information is specified, then the initial value
problem is well-posed and it can be solved using only a {\em
finite number} of local initial data. Also, motivated by some
intriguing work by Dragovich and Aref'eva-Volovich on 
cosmology, we solve explicitly field equations of the form
\begin{equation*}
\zeta(\partial_t + h) \phi = J(t) \; , \; \; \; t \geq 0 \; ,
\end{equation*}
in which $\zeta$ is the Riemann zeta function and $h > 1$. 
Finally, we remark that the $L^2$ case of our general theory 
allows us to give a precise meaning to the
often-used interpretation of $f(\partial_t)$ as an operator defined by a power 
series in the differential operator $\partial_t$.
\end{abstract}

\section{Introduction}\label{sec:Introduction}

Equations with an infinite number of derivatives have appeared recently
as field equations of motion in particle physics \cite{Mo1}, string 
theory \cite{BCK,D1,D2,EW,M,Ta,V,VV,W}, and (quantum) gravity and
cosmology \cite{AV,AV2,B,BBC,BK1,BK2,BGKM,Cal1,CalMo,Ko,Moeller}.
For instance, an important equation in this class is
\begin{equation} \label{padic-1}
p^{a\, \partial^2_t}\phi = \phi^p \; , \; \; \; \; \; a > 0 \; ,
\end{equation}
where $p$ is a prime number. Equation (\ref{padic-1}) describes the dynamics 
of the open $p$-adic string for the scalar tachyon field (see 
\cite{AV,BK2,D,Moeller,V,VV} and references therein) and it can be understood, 
at least formally, as an equation in an 
infinite number of derivatives if we expand the entire
function (called below the ``symbol" of the equation)
$$f(s):=p^{as^2}=e^{as^2\, log(p)}$$
as a power series around zero and we replace $s$ for $\partial_t$.
This equation has been studied by Vladimirov via integral
equations of convolution type in \cite{V,VV} (see also
\cite{AV,Moeller}), and it has been also noted that in the limit
$p \rightarrow 1$, Equation (\ref{padic-1}) becomes the {\em
local} logarithmic Klein-Gordon equation \cite{BG,GS}. Moreover,
Dragovich, see \cite{D}, has used (the Lagrangian formulation of) Equation (\ref{padic-1}) and the
Euler product for the Riemann zeta function
$$\zeta(s):=\sum_{n=1}^{\infty}\dfrac{1}{n^s}\; , \quad \quad Re(s)>1\; ,$$
see \cite{KaVo}, to deduce a string theory field equation of the form
\begin{equation} \label{zeta}
\zeta \left(- \frac{1}{m^2} \, \partial_t^2 + h \right) \phi = U(\phi) \; ,
\end{equation}
in which $m,h$ are constant real numbers and $U$ is a nonlinear
function of $\phi$. Equation (\ref{zeta}) is called a zeta-nonlocal field equation
in \cite{D}. 

In this paper we focus on analytic properties of {\em linear} equations
of the form
\begin{equation} \label{abs222}
f(\partial_t) \phi = J(t) \; , \quad t \geq 0 \; ,
\end{equation}
in which $f$ is an analytic function and $J$ belongs to an
adequate class of functions to be specified in Section 2. We also
investigate carefully the formulation, existence and solution of
initial value problems, and we explain precisely in what sense
Equation (\ref{abs222}) really {\em is} an ordinary differential
equation in infinitely many derivatives. Motivated mainly by
Dragovich's work (see also \cite{AV2}), we endeavor to develop a
theory for $f(\partial_t)$ able to deal with the case in which $f$
is, for instance, the Riemann zeta function. We think this
flexibility is important: it does not appear to be straightforward
to use the previous works \cite{GPR_JMP,GPR_CQG} for equations of
the form $\zeta(\partial_t + h) \phi = J$, as we explain in
Section 2.

We remark that linear equations in infinitely many derivatives appeared
in mathematics already in the final years of the XIX Century, see for
instance \cite{Bou,Car} and further references in \cite{BK1}. 
It appears to us, however, that a truly fundamental stimulus for
their study has been the realization that linear and nonlinear
nonlocal equations play an important role in contemporary physical
theories, as witnessed by \cite{PU} and the more recent papers
cited above.

Now, a serious problem for the development of a rigorous --yet
flexible enough to allow for explicit computations-- theory for
nonlocal equations  has been the difficulties inherent in the
understanding of the initial value problem for equations such as
(\ref{padic-1}) and (\ref{abs222}). An interpretative difficulty
considered already in the classical paper \cite{EW} by Eliezer and
Woodard is the frequently stated argument (see for instance
\cite{B,Moeller} and the rigorous approach of \cite{CPR}) that if
an $n$th order ordinary differential equation requires $n$ initial
conditions, then the ``infinite order'' equation
$F(t,q',q^{''},\cdots , q^{(n)},\cdots)=0$ requires infinitely
many initial conditions, and therefore the solution $q$ would be
determined a priori (via power series) without reference to the
actual equation.

In order to deal with this difficulty, we investigate initial
value problems from scratch. Our approach is to emphasize the role
played by the Laplace transform, in the spirit of
\cite{GPR_JMP,GPR_CQG} and the interesting papers
\cite{AV,AV2,BK1,BK2}. Our solution (Section 3.2 below) is that if
an a priori data directly connected with our interpretation of
$f(\partial_t)$ is specified (Equation (42) below) then the
initial value problem is well-posed and it requires only a {\em
finite number} of initial conditions. 
Due to the preeminence of Laplace transform in our theory, it is
natural for us to speak indistinctly of ``nonlocal equations'' or
``equations in an infinite number of derivatives'', and we do so
hereafter. 

We finish this Introduction with the remark that our
work is not necessarily part of the classical theory of
pseudo-differential operators for two reasons. First, the symbol $f$ appearing in
Equation (\ref{abs222}) is (in principle) an arbitrary analytic
function, and therefore the ``operator" $f(\partial_t)$ may be
beyond the reach of classical tools such as the ones appearing in
\cite{H}. Second, we aim at solving nonlocal equations on the
semi-axis $t \geq 0$, not on the real line where perhaps we could
use Fourier transforms as in the intriguing papers \cite{AV2,D},
and/or classical pseudo-differential analysis \cite{H}; motivation for the study
of nonlocal equations on a semi-axis comes, for example, from cosmology \cite{AV,MN}: 
classical versions of Big Bang cosmological models contain singularities at the beginning 
of time, and therefore the time variable appearing in the field
equations should vary over a half line.

\smallskip

This paper is organized as follows. In Section 2 we consider a
rigorous interpretation of nonlocal ordinary differential
equations via Laplace transform considered as an operator between
appropriate Banach spaces. In Section 3 we investigate (and propose a method
for the solution of) initial value problems for linear equations
of the form (\ref{abs222}). In Section 4 we solve a class of zeta
nonlocal equations and in Section 5 we provide an $L^2$-theory for
operators $f(\partial_t)$. Finally in Section 6 we discuss some
problems which are beyond the reach of the theory presented here;
they will be considered in the companion article \cite{CPR_Borel}.

\section{Linear nonlocal equations}

In this section we introduce a rigorous and computationally useful
framework for the analytic study of nonlocal equations, including
the possibility of setting up meaningful initial value problems.
We consider linear nonlocal equations of the form
\begin{equation} \label{lor_linear}
  f(\partial_t)\phi(t) - J(t) =0 \; , \; \; \; t \geq 0\; ,
\end{equation}
in which $f$ is an analytic function.

\subsection{Preliminaries}

We begin by fixing our notation and stating some preliminary facts; 
our main source is the classical treatise \cite{Doetsch} by G. Doetsch.

We recall that a function $g$ belongs to $L^1_{loc}(\mathbb{R}_+)$, in which
$\mathbb{R}_+$ is the interval $[0,\infty)$, if and only if $$\int_K |g(t)|dt$$
exists for any compact set $K \subseteq \mathbb{R}_+$.

\begin{defi}
A function $g : \mathbb{R} \rightarrow \mathbb{C}$ in $L^1_{loc}(\mathbb{R}_+)$ 
belongs to the class $\mathcal{T}_a$ if and only if $g(t)$ is
identically zero for $t < 0$ and the integral
\[
\int_0^\infty e^{- s\,t}\,g(t)\,dt
\]
converges absolutely for $Re(s) > a$.
\end{defi}

$\mathcal{T}_a$ is simply a vector space; we do not need to endow
it with a topology. The {\em Laplace transform} of a function $g :
\mathbb{R} \rightarrow \mathbb{C}$ in $\mathcal{T}_a$ is the
integral
\[
{\cal L}(g)(s) = \int_0^\infty e^{- s\,t}\,g(t)\,dt \; , \quad \;
Re(s) > a \; .
\]
As proven in \cite[Theorem 3.1]{Doetsch}, if $g \in
\mathcal{T}_a$, then ${\cal L}(g)(s)$ converges absolutely and
uniformly for $Re(s) \geq x_0 > a$ (a finer statement is in
\cite[Theorem 23.1]{Doetsch}). It is also known (see \cite[Theorem
6.1]{Doetsch}) that the Laplace transform ${\cal L}(g)$ is an
analytic function for $Re(s) > a$. We also record the following:

\begin{prop} \label{proplaplace}
~
\begin{enumerate}
\item Let us fix a function $g \in \mathcal{T}_a$. If the
derivative $D^k g (t)$ belongs to $\mathcal{T}_a$ for $0 \leq k
\leq n-1$, then $D^n g(t)$ belongs to $\mathcal{T}_a$ and
\[
{\cal L}(D^n g)(s) = s^n\, {\cal L}(g)(s) - \sum_{j=1}^n (D^{j-1}
g)(0) s^{n-j} \; .
\]
\item If $g \in \mathcal{T}_a$ is a continuous function, then for any $\sigma > a$ we have
the identity
\begin{equation} \label{inv0}
g (t) = \frac{1}{2 \pi i}\,\int_{\sigma - i \infty}^{\sigma + i
\infty} e^{s\,t}\, {\cal L}(g)(s)\, ds \; 
\end{equation}
for all $t \geq 0$. 
\end{enumerate}
\end{prop}
The second part of this proposition is in  \cite[Theorem
24.4]{Doetsch}. The class $\mathcal{T}_a$ appears to be the
simplest class of functions on which the Laplace transform is
defined and for which there exists an inversion formula
(\ref{inv0}). The inverse Laplace transform is denoted by $\mathcal{L}^{-1}$.

We recall that we are interested in a theory able to deal with
(\ref{lor_linear}) in which, for instance, $f(s) = \zeta(s+h)$.
This is not a trivial requirement: let us suppose that we are
interested in applying Theorem 3.1 of \cite{GPR_CQG} on classical
initial value problems to the equation $\zeta(\partial_t+h)\phi(t)
= J(t)$. Then, we would need to check that
$\mathcal{L}(J)(s)/\zeta(s + h)$ belongs to the Widder space
$$C^\infty_W (\omega , \infty ) = \left \{ r : (\omega , \infty)
\rightarrow \mathbb{C}\; / \; \| r \|_W = \displaystyle\sup_{n \in
\mathbb{N}_0} \displaystyle \sup_{ s > \omega } \left|
\frac{(s-\omega)^{n+1}}{n !}\, r^{(n)}(s) \right| < \infty \right
\}\; ,$$ in which $\omega > 0$ and $r^{(n)}(s)$ denotes the $n$th
derivative of $r(s)$. Certainly, such a check does not look
straightforward. It is important therefore to develop an alternative approach
to nonlocal equations.

The class $\mathcal{T}_a$ appears to be a good starting point. Regretfully, because 
of reasons to be explained in Section 4, it is not adequate for defining and solving initial
value problems. We have found that a better alternative is to consider Laplace transform as a correspondence 
from Lebesgue spaces $L^p(\mathbb{R}_+)$ into Hardy spaces $H^q(\mathbb{C}_+)$, which we now define.

We write $\mathbb{C}_+$ for the right half-plane $\{s\in \mathbb{C}\; : Re(s)>0
\}$. The space $L^p(\mathbb{R}_+)$, $1 \leq p < \infty$, is the Lebesgue
space of measurable functions $\phi$ on $[0,\infty)$ such that
$$||\phi ||_{L^p(\mathbb{R}_+)}:=\left( \int_0^{\infty}|\phi(x)|^pdx \right)^{\frac{1}{p}} \; <\; \infty\; ,$$
and the $q$th Hardy space $H^q(\mathbb{C}_+)$ is the space of all functions $\Phi$ 
which are analytic 
on $\mathbb{C}_+$ and such that the integral $\mu_q(\Phi,x)$ given by
$$\mu_q(\Phi,x):= \left( \frac{1}{2\pi}\int_{-\infty}^{\infty}|\Phi(x+iy)|^q dy \right)^{\frac{1}{q}} \; ,$$
is uniformly bounded for $x > 0$. We note that $H^q(\mathbb{C}_+)$ 
becomes a Banach space with the norm 
$||\Phi||_{H^q(\mathbb{C}_+)}:=\sup_{x>0}\mu_q(\Phi,x)$.

The following classic Representation theorem was first presented by Doetsch in 
\cite[pp. 276 and 279]{Doetsch2}:

\begin{tw}\label{twDoe}
(Doetsch's Representation theorem)
\begin{itemize}
\item [(i)] If $\phi \in L^p(\mathbb{R}_+)$, where $1<p\leq 2$,
and $\Phi=\mathcal{L}(\phi)$, then $\Phi \in H^{p'}(\mathbb{C}_+)$
with $\frac{1}{p}+ \frac{1}{p'}=1$. Moreover if $x>0$ there exists
a positive constant $C(p)$ such that:
$$ \mu_{p'}(\Phi,x)\leq C(p) \left( \int_0^{\infty}e^{-pxt}|\phi(t)|^p dt \right)^{\frac{1}{p}}\; .$$
\item[(ii)] If $\Phi \in H^{p}(\mathbb{C}_+)$, where $1<p\leq 2$,
then there exists $\phi \in L^{p'}(0,\infty)$ with $\frac{1}{p}+
\frac{1}{p'}=1$ suth that $\Phi=\mathcal{L}(\phi)$. The function
$\phi$ is given by the inversion formula
$$\phi(t):=\lim_{v\to \infty}\frac{1}{2\pi}\int_{-v}^{v} e^{(\sigma + i\eta)t}\Phi(\sigma + i\eta ) d\eta\; , \; \; 
\quad \sigma \geq 0\; ,$$
in which the limit is understood in $L^{p'}(\mathbb{R}_+)$, 
and for $x>0$ there exists a positive constant $K(p)$ such that
$$\left( \int_0^{\infty}e^{-p' x t}|\phi(t)|^{p'} dt \right)^{\frac{1}{p'}}\leq K(p)\mu_{p}(\Phi,x).$$
\end{itemize}
\end{tw}

We remark that the restrictions on $p$ appearing in both parts of the theorem imply that neither 
{\it (i)} is the converse of {\it (ii)} nor {\it (ii)} is the converse of {\it (i)}, except in the 
case $p=p'=2$.  This fact reflects itself in the enunciates of
our main theorems (Theorems \ref{main_thm} and \ref{ivp2}). 
If $p=p'=2$, Doetsch's Representation theorem is the important
Paley-Wiener theorem, which states that the Laplace transform is a
unitary isomorphism between $L^2(\mathbb{R}_+)$ and
$H^2(\mathbb{C}_+)$; a precise statement is in Section 5, after \cite{Har,KoYo}. 
We also note that there exist generalizations of Theorem \ref{twDoe}: it has been observed that 
the Laplace transform determine correspondences between appropriate
weighted Lebesgue and Hardy spaces, see \cite{BeHeJo,Rooney} and references therein, and
there also exists a representation theorem for functions $\Phi$ which
decay to zero in any closed half-plane $Re(s)\geq \delta
>a,\; a\in \mathbb{R}$ and satisfy the additional requirement that
$\Phi(\sigma +i (\cdot))\in L^1(\mathbb{R})$ (roughly speaking, an 
``$L^1$-case" of Theorem \ref{twDoe}, see \cite{Chu,Doetsch}). Interestingly, as we explain in
Section 3, in this ``$L^1$-situation" we can formulate theorems on existence and uniqueness of
solutions for nonlocal linear equations but regretfully, we cannot develop a
meaningful theory of initial value problems. 

The Doetsch Representation theorem is our main tool for the understanding
of nonlocal equations.

\subsection{The operator $f(\partial_t) : L^p(\mathbb{R}_+) \longrightarrow H^q(\mathbb{C}_+)$}

As a motivation, we calculate $f(\partial_t) \phi$ formally: we take a function $f$ analytic around zero, and suppose that
$\phi \in L^{p'}(0,\infty)$ is smooth. Let us write
\[
 f(\partial_t) \phi = \sum_{n=0}^\infty \frac{f^{(n)}(0)}{n!}\, \partial_t^n \phi \; .
\]
We note that this expansion is actually a rigorous result if $f$ extends
to an entire function and $\phi$ is a $f$-analytic vector for $\partial_t$ on the space of square integrable 
functions $L^2(\mathbb{R}_+)$ (see Section 5 below); also, this expansion is a formal 
theorem if $f$ is an entire function and $\phi$ is an entire function of exponential type (see \cite{CPR}; 
a generalization of this result appears in the companion paper \cite{CPR_Borel}). 
Standard properties of
the Laplace transform (see Proposition \ref{proplaplace} or \cite{Doetsch}) yield, formally,
\begin{eqnarray}
\mathcal{L} ( f(\partial_t)\phi )(s) & = &
\sum_{n=0}^\infty \frac{f^{(n)}(0)}{n!}\, \mathcal{L}(\partial_t^n\phi) \nonumber \\
& = & \sum_{n=0}^\infty \frac{f^{(n)}(0)}{n!}\, (s^n \mathcal{L}(\phi) - s^{n-1}\phi(0) -
                             s^{n-2}\phi'(0) - \dots - \phi^{(n-1)}(0) ) \nonumber \\
& = & \sum_{n=0}^\infty \frac{f^{(n)}(0)}{n!}\, \left(s^n \mathcal{L}(\phi) - \sum_{j=1}^n
                              s^{n-j} \phi^{(j-1)}(0) \right) \nonumber \\
& = & f(s) \mathcal{L}(\phi)(s) - \sum_{n=1}^\infty \sum_{j=1}^n \frac{f^{(n)}(0)}{n!}\,
                              s^{n-j} \phi^{(j-1)}(0) \; . \label{aux1}
\end{eqnarray}
If we define the formal series
\begin{equation}  \label{in_con111}
r(s) =  \sum_{n=1}^\infty \sum_{j=1}^n \frac{f^{(n)}(0)}{n!}\,
d_{j-1} \, s^{n-j} \; ,
\end{equation}
in which $d = \{ d_{j} : j \geq 0 \}$ is a sequence of complex numbers, then we can write (\ref{aux1}) as
\begin{equation} \label{xxx}
 \mathcal{L} ( f(\partial_t)\phi )(s) = f(s) \mathcal{L}(\phi)(s) - r(s) \; ,
\end{equation}
where $r(s)$ is given by (\ref{in_con111}) with $d_j = \phi^{(j)}(0)$.

The following lemma asserts that for some choice of the sequence
$\{d_{j-1}\}_{j\geq 1}$, the formal series defined in (\ref{in_con111}) is in fact an analytic function.
\begin{lem}\label{r_conv}
Let $R_1>1$ be the maximum radius of convergence of the Taylor series 
$f_T(s) = \sum_{n=0}^\infty \frac{f^{(n)}(0)}{n!}\, s^n \; .$ Set $0<R<1$ and suppose that the series
\begin{equation}\label{eq_NB}
 \sum_{j=1}^{\infty}d_{j-1} \dfrac{1}{s^j} \;
\end{equation}
is uniformly convergent on compact sets for $|s|>R$. Then the series
$(\ref{in_con111})$ is analytic on the disk $|s|<R_1$.
\end{lem}
\begin{proof}
Let us write the series (\ref{in_con111}) in the following form
\[
 r(s) =  \sum_{n=1}^\infty \frac{f^{(n)}(0)}{n!} s^n\sum_{j=1}^n \,
d_{j-1} \, s^{-j} \; ,
\]
and define $r_N(s)$ as the partial sum
$$r_N(s):= \sum_{n=1}^N \frac{f^{(n)}(0)}{n!} s^n\sum_{j=1}^n \,
d_{j-1} \, s^{-j} \; .$$
Let $K\subset \{s: R<|s|<R_1\}$ be a compact set; there exists a positive constant $L_K$ 
such that for any $s\in K$ we have:
$$\, \sum_{j=1}^{\infty} \left| \frac{
d_{j-1}} { s^{j}} \right| \leq L_K \; ;$$
also, given $\epsilon>0$, there exist $N_0=N_0(\epsilon)$ such that
$$\sum_{n=N_0+1}^{\infty} \left|\frac{f^{(n)}(0)}{n!} s^n\right| <\epsilon\; ; $$

\noindent from these inequalities we deduce that
$$|r(s)-r_N(s)|=\left|\sum_{n=N+1}^{\infty} \frac{f^{(n)}(0)}{n!} s^n \sum_{j=1}^{n}  \frac{
d_{j-1}} { s^{j}} \right| < \epsilon L_K\; ,$$
for any $N\geq N_0$ and uniformly for  $s \in K$. Therefore, the partial sums $r_N(s)$ 
converge uniformly to $r(s)$ on $K$; 
since the compact subset $K$ is arbitrary, we have that $r(s)$ is analytic in $\{s: R<|s|<R_1\}$. 
Now we prove that $r(s)$ is analytic in $|s|\leq R$. We write $r(s)$ as
\[
 r(s) =  \sum_{n=1}^\infty \frac{f^{(n)}(0)}{n!}P_n(s) \; ,
\]
in which,
\[
P_n(s):=\sum_{j=1}^n \,
d_{j-1} \, s^{n-j}\, =d_0S_{n-1} + d_1S_{n-2}+\cdots d_{n-1}S_{0}\; ,
\]
where $S_n=\sum_{k=1}^{n}s^k$. Then, the convergence of the series $\sum_{j=1}^{\infty}d_{j-1}$
and the inequalities $R<1<R_1$ imply that the following three assertions hold:
\begin{itemize}
\item $\displaystyle
\sum_{n=1}^\infty \left| \frac{f^{(n)}(0)}{n!}-\frac{f^{(n+1)}(0)}{(n+1)!} \right|  < \infty\,$,
\item
$\lim_{n\to \infty} \frac{f^{(n)}(0)}{n!} =0\,$,
\item
$P_n(s)$ is uniformly bounded on $|s|\leq R\,$. 
\end{itemize}
The result now follows from \cite[Theorem 5.1.8]{Hi}.
\end{proof}

As a corollary we recover a previously known lemma, see \cite[Lemma 2.1]{GPR_JMP}:

\begin{cor}
Let $f$ be an entire function. Set $R<1$ and suppose that the series
\[
 \sum_{j=1}^{\infty}d_{j-1} \dfrac{1}{s^j} \; ,
\]
is convergent for $|s|>R$. Then the series $(\ref{in_con111})$ is an entire function.
\end{cor}
\begin{rem}
There exists a large class of series satisfying conditions of Lemma
$\ref{r_conv}$. Indeed, let $r>0$ and denote by $Exp_{r}(\mathbb{C})$ the space
of entire functions of exponential type $\tau <r$ (see {\rm \cite{CPR,U}} and references therein). 
It is well know that if $\phi \in Exp_{r}(\mathbb{C})$ and it is of exponential type $\tau<r$, 
then its Borel transform 
$$\mathcal{B}(\phi)(s)= \sum_{j=0}^{\infty} \dfrac{\phi^{j}(0)}{s^{j+1}}$$
converges uniformly on $|s|>\tau$ (see again {\rm \cite{CPR,U}}). In our case, if 
$\phi \in Exp_{1}(\mathbb{C})$, 
its Borel transform $\mathcal{B}(\phi)(s)$ is precisely $(\ref{eq_NB})$ for $d_{j}=\phi^{j}(0)$. 
We will come back to 
the class of functions of exponential type in our forthcoming paper {\rm \cite{CPR_Borel}}.
\end{rem}

Motivated by the previous computations and Doetsch's representation theorem, we make the 
following definition, generalizing \cite{GPR_JMP,GPR_CQG}:

\begin{defi} \label{def0}
Let $f$ be an analytic function on a region which contains the half-plane
$\{ s \in \mathbb{C} : Re(s) > 0 \}$, and let $\mathcal{H}$ be the space of all 
$\mathbb{C}$-valued functions 
on $\mathbb{C}$ which are analytic on (regions of) $\mathbb C$. We fix $p$ and $p'$ such that 
$1<p\leq 2$ and 
$\frac{1}{p}+\frac{1}{p'}=1$, and we consider the subspace $D_f$ of
$L^{p'}(0,\infty) \times \mathcal{H}$ consisting of all the pairs $(\phi , r)$ such that
\begin{equation} \label{tr0}
\widehat{(\phi , r)} = f\,{\mathcal L}(\phi) - r
\end{equation}
belongs to the class $H^{p}(\mathbb{C}_+)\,$.
The domain of $f(\partial_t)$ as a linear operator from the product $L^{p'}(0,\infty) \times \mathcal{H}$
into $L^{p'}(0,\infty)$ is $D_f$. If $(\phi , r) \in D_f$ then we define
\begin{equation} \label{tr1}
f(\partial_t)\,(\phi , r) = {\mathcal L}^{- 1} (\,\,\widehat{(\phi,r)} \,\,) =
{\mathcal L}^{- 1} ( f\,{\mathcal L}(\phi)
- r ) \; .
\end{equation}
\end{defi}

\smallskip

We note that this definition will be further generalized in \cite{CPR_Borel} by means of the 
Borel transform, 
in order to capture an even larger class of symbols $f$. However, as we will see below, the 
foregoing interpretation 
for $f(\partial_t)$ already gives us a satisfactory way to deal with the initial value problem 
for linear nonlocal equations in many interesting cases.

\subsection{Linear nonlocal equations}
In this subsection we solve the nonlocal equation
\begin{equation} \label{lin_gen_0}
f(\partial_t)(\phi , r) = J \; , \quad \quad J \in L^{p'}(0,\infty)\; ,
\end{equation}
in which we are using the above interpretation of $f(\partial_t)$. We assume
hereafter that a suitable function $r \in \mathcal{H}$ has been fixed; consequently,
we understand Equation (\ref{lin_gen_0}) as an
equation for $\phi \in L^{p'}(0,\infty)$ such that $(\phi , r) \in D_f$.
We simply write $f(\partial_t)\phi = J$ instead of (\ref{lin_gen_0}). First of all, 
we formalize what we mean by a solution:

\begin{defi} \label{main_def}
 Let us fix a function $r \in \mathcal{H}$. We say that
 $\phi \in L^{p'}(0,\infty)$
is a solution to the equation $f(\partial_t)\phi = J$ if and only if
\begin{enumerate}
 \item $\widehat{\phi} = f\,\mathcal{L}(\phi) - r \in H^p(\mathbb{C}_+)\,$;
 $($ i.e., $(\phi,r) \in D_f$ $)$;
 \item $f(\partial_t)(\phi) = \mathcal{L}^{-1} (\widehat{(\phi,r)}) =
                             \mathcal{L}^{-1}(f\,\mathcal{L}(\phi) - r) = J\,$.
\end{enumerate}
\end{defi}

Our main theorem on existence and uniqueness of the solution to the
linear problem (\ref{lin_gen_0}) is the following abstract result:

\begin{tw} \label{main_thm}
Let us fix a function $f$ which is analytic in a region $D$
which contains the half-plane $\{s \in \mathbb{C} : Re(s) > 0\}$. We also fix $p$ 
and $p'$ such that $1<p\leq 2$ and $\frac{1}{p}+\frac{1}{p'}=1$, and we consider a
function $J \in L^{p'}(\mathbb{R}_+)$ such that 
$\mathcal{L}(J) \in H^p(\mathbb{C}_+)$.
We assume that the function $({\cal L}( J ) + r)/f$ is in the space 
$H^p(\mathbb{C}_+)$. Then, the linear equation
\begin{equation} \label{lin_gen}
f(\partial_t)\phi = J
\end{equation}
can be uniquely solved on $L^{p'}(0,\infty)$.
Moreover, the solution is given by the explicit formula
\begin{equation} \label{sol_lor}
\phi = {\cal L}^{-1} \left( \frac{{\cal L}( J ) + r}{f} \, \right)
\; .
\end{equation}
\end{tw}
\begin{proof}
We set $\phi = {\cal L}^{-1} \left(({\cal L}( J ) + r)/f\,\right)\,$.
Since $\mathcal{L}(J) \in H^p(\mathbb{C}_+)$, it follows that the pair $(\phi , r)$
is in the domain $D_f$ of the operator $f(\partial_t)$: indeed, an easy calculation using 
Theorem \ref{twDoe} shows that $\widehat{\phi}={\cal L}( J )$, which is an element 
of $H^p(\mathbb{C}_+)$ by hypothesis.  We can then check (using Theorem \ref{twDoe}
again) that $\phi$ defined by (\ref{sol_lor}) is a solution of (\ref{lin_gen}).

We prove uniqueness using Definition \ref{main_def}:
let us assume that $\phi$ and $\psi$ are solutions to Equation
$(\ref{lin_gen})$. Then, item $2$ of Definition $\ref{main_def}$ implies
$f\,\mathcal{L}(\phi - \psi) =0$ on $\{s \in \mathbb{C} : Re(s) > a\}$. Set 
$h = \mathcal{L}(\phi - \psi)$ 
and suppose that $h(s_0) \neq 0$ for $s_0$ in the half-plane just defined. By
analyticity, $h(s) \neq 0$ in a suitable neighborhood $U$ of $s_0$.
But then $f=0$ in $U$, so that (again by
analyticity) $f$ is identically zero.
\end{proof}

\begin{rem}
Interestingly, the above theorem on solutions to Equation $(\ref{lin_gen})$ is 
a fully rigorous version of theorems stated long ago, see for instance the classical
papers {\rm \cite{Bou,Car}} by Bourlet and Carmichael, and also the recent work 
{\rm \cite{BK1}} in which further references appear. We recover a theorem proven by 
Carmichael in {\rm \cite{Car}} in Corollary $\ref{ic1}$ below. 
\end{rem}

In Section 3 we impose further conditions on $J$ and $f$
which assure us that (\ref{sol_lor}) is smooth at $t = 0$, and we
use these conditions to study the initial value problem for
(\ref{lin_gen}).

We recall that we have fixed an analytic function $r$. Hereafter we assume the
natural decay condition
\begin{equation}
\Big|\frac{r(s)}{f(s)}\Big| \leq \displaystyle\frac{C}{|s|^q}
\label{growth0}
\end{equation}
for $|s|$ sufficiently large and some real number $q > 0$ (see for instance
\cite{Doetsch}), and we examine
three special cases of Theorem \ref{main_thm}: (a) the function
$r/f$ has no poles; (b) the function $r/f$ has a finite number of
poles; (c) the function $r/f$ has an infinite number of poles.

\smallskip

\begin{cor} \label{ic}
Assume that the hypotheses of Theorem $\ref{main_thm}$ hold, that
${\cal L}( J )/f$ is in $H^p(\mathbb{C}_+)$, and that
$r/f$ is an entire function such that $(\ref{growth0})$ holds.
Then, solution $(\ref{sol_lor})$ to Equation $(\ref{lin_gen})$ is
simply $\phi = {\cal L}^{-1} \left( \frac{{\cal L}( J )}{f} \,
\right)$.
\end{cor}
\begin{proof}
The proof is analogous to the proof of Corollary 2.1 of \cite{GPR_CQG}.
\end{proof}

\smallskip

\begin{cor} \label{ic1}
Assume that  the hypotheses of\, Theorem $\ref{main_thm}$ hold,
and that $\displaystyle r/f$\, has a finite number of
poles $\omega_i$ ($i=1, \dots, N$) of order $\,r_i$ to the left of
$Re(s)=0$. Suppose also that ${\cal L}( J )/f$ is in
$H^p(\mathbb{C}_+)$, and that the growth condition $(\ref{growth0})$ holds.
Then, the solution $\phi \in L^{p'}(\mathbb{R}_+)$ given by $(\ref{sol_lor})$ 
can be represented in the form
\begin{equation} \label{car}
\phi(t) = \frac{1}{2 \pi i} \int_{\sigma - i \infty}^{\sigma + i\infty} e^{s\, t} \left( \frac{{\cal L}( J )}{f} \right)(s) \, ds +
\sum_{i=1}^N P_i(t)\,e^{\omega_i t} \; \; ,\; \quad \sigma>0 \; ,
\end{equation}
in which $P_i(t)$ are polynomials of degree $\,r_i - 1$.
\end{cor}

\begin{proof}
We first notice that the quotient $r/f\in H^p(\mathbb{C}_+)$, since 
$({\cal L}( J ) + r)/f$ and ${\cal L}( J )/f$ are in $H^p(\mathbb{C}_+)$, 
by assumption. Thus, 
the general solution (\ref{sol_lor}) can be computed via the inversion
formula for the Laplace transform appearing in Theorem \ref{twDoe} (see also 
\cite{Doetsch}). We obtain the following equation in $L^{p'}(\mathbb{R}_+)$:
\begin{equation}\label{inversion}
\phi(t) =
 \frac{1}{2 \pi i} \int_{\sigma - i \infty}^{\sigma + i\infty} e^{s\, t} \left( \frac{{\cal
L}( J )}{f} \right)(s)ds  + \frac{1}{2 \pi i} \int_{\sigma - i \infty}^{\sigma + i\infty} e^{s\, t} \left( \frac{r}{f} \right)(s) \, ds \; , \quad \sigma >0 \; .
\end{equation}
At this level of generality, all we can say about the first summand is that it 
belongs to $L^{p'}(\mathbb{R}_+)$. On the other hand, the growth condition on 
$r/f$ implies that we can evaluate the second summand via calculus of
residues, see \cite[Theorem 25.1 and Section 26]{Doetsch}.
The second integral in the right hand side of (\ref{inversion}) is 
$$
\frac{1}{2 \pi i} \int_{\sigma - i \infty}^{\sigma + i\infty} e^{s\,
t} \left( \frac{r}{f}\right)(s) \, ds=\sum_{i=1}^{N}res_{\,i}(t)\; ,
$$
in which $res_{\,i}(t)$ denotes the residue of $\displaystyle
\frac{r(s)}{f(s)}$ at $\omega_i\,$. In order to compute
$res_{\,i}(t)$ we use the Laurent
expansion of $r/f$ around the pole $\omega_i\,$, that is
\begin{equation} \label{pole-1}
\frac{r(s)}{f(s)}=\frac{a_{1,i}}{(s-\omega_i)}
+\frac{a_{2,i}}{(s-\omega_i)^{2}}+\cdots+\frac{a_{r_i,i}}{(s-\omega_i)^{r_{i}}} + h_i(s) \; ,
\end{equation}
where $h_i$ is an analytic function inside a closed curve around
$\omega_i\,$. We multiply (\ref{pole-1}) by $e^{ts}/2\pi i$ and use
Cauchy's integral formula. We obtain 
(cf. \cite[{\em loc. cit.}]{Doetsch})
\[
res_{\,i}(t) = P_i(t)\;e^{\omega_i t} \; ,
\]
where $P_i(t)$ is the polynomial of degree $r_i - 1$ given by
\begin{equation}
 P_i(t) = a_{1,i} + a_{2,i}\, \frac{t}{1!} + \cdots + a_{r_i,i}\, \frac{t^{r_i - 1}}{(r_i -1)!} \; ,
\end{equation}
and the result follows. We note that $res_i(t)$ is indeed in $L^{p'}(\mathbb{R}_+)$ 
since the poles $\omega_i$ lie to the left of $\mathbb{C}_+$.
\end{proof}

Formula (\ref{car}) for the solution $\phi(t)$ is
precisely Carmichael's formula appearing in \cite{Car}, as quoted
in \cite{BK1}. Explicit formulae such as
(\ref{car}) are crucial for the study of the initial value problem
we carry out in Section 3.

The last special case of Theorem \ref{main_thm} is when the
function $r(s)/f(s)$ has an infinite number of isolated poles
$\omega_i$ located to the left of $Re(s) = a$, and such that
$|\omega_0| \leq |\omega_1| \leq |\omega_2| \leq \cdots$.

\begin{cor} \label{ic2}
Assume that  the hypotheses of\,\, Theorem $\ref{main_thm}$ hold, that 
${\cal L}( J )/f$ 
is in $H^p(\mathbb{C}_+)$, that the quotient $r/f$\, has an infinite number of
poles $\omega_n$ of order $\,r_n$ to the left of
$Re(s)=0$, and that $|\omega_n| \leq |\omega_{n+1}|$ for\, $n \geq 1$. 
Suppose that there exist curves $\sigma_n$ in the half-plane
$Re(s) \leq 0$ satisfying the following:
\begin{itemize}
\item[~] The curves $\sigma_n$ connect points $+ i b_n$ and $- i b_n\,$, where
the numbers $b_n$ are such that the closed curve formed by $\sigma_n$ together with
the segment of the line $Re(s) = 0$ between the points $+ i b_n$ and 
$- i b_n\,$, encloses exactly the first $n$ poles of $r(s)/f(s)$ and, moreover, 
$\lim _{n \rightarrow \infty} b_n = \infty$.
\end{itemize}
If the condition 
$$\lim_{n\to \infty} \int_{\sigma_n}e^{st} \left( \frac{r}{f} \right)(s) ds =0 $$
(almost everywhere in $t$) holds, then the solution $\phi \in L^{p'}(\mathbb{R}_+)$ 
given by $(\ref{sol_lor})$, can be represented in the form
\begin{equation} \label{car-1}
\phi(t) = \frac{1}{2 \pi i} \int_{\sigma - i \infty}^{\sigma +
i\infty} e^{s\, t} \left( \frac{{\cal L}( J )}{f} \right)(s) \, ds +
\sum_{n=1}^\infty P_n(t)\,e^{\omega_n t} \; , \quad \sigma > 0\; ,
\end{equation}
in which $P_n(t)$ are polynomials of degree $\,r_n - 1$.
\end{cor}
\begin{proof}
As in the previous corollary, we can write solution (\ref{sol_lor}) in the form 
\begin{equation}\label{inversion-1}
\phi(t) =
 \frac{1}{2 \pi i} \int_{ - i \infty}^{ + i\infty} e^{s\, t} \left( \frac{{\cal
L}( J )}{f} \right)(s)ds  + \frac{1}{2 \pi i} \int_{ - i \infty}^{ + i\infty} e^{s\, t} \left( \frac{r}{f} \right)(s) \, ds \; , 
\end{equation}
as an equation in $L^{p'}(\mathbb{R}_+)$. Since the integral in the second summand 
converges in $L^{p'}(\mathbb{R}_+)$, there exists an increasing sequence of real
numbers $b_n$ such that
$$
\lim_{n \rightarrow \infty} \frac{1}{2 \pi i} \int_{- i b_n}^{ + ib_n} e^{s\, t} 
\left( \frac{r}{f} \right)(s) \, ds = \mathcal{L}^{-1}(\frac{r}{f})(t)\; , \; \quad t \geq 0\; , $$
almost everywhere in $t$, in which $\lim$ is being taken in the metric topology of 
$\mathbb{C}$. Following Doetsch \cite[p. 170]{Doetsch}, we compute using residues:
$$
\frac{1}{2 \pi i} \int_{ - i b_n}^{ + ib_n} e^{s\, t} 
\left( \frac{r}{f} \right)(s) \, ds + \frac{1}{2 \pi i} \int_{\sigma_n} e^{s\, t} 
\left( \frac{r}{f} \right)(s) \, ds = \sum_{j=1}^n res_j(t) \; .
$$
These residues are calculated as in the previous corollary; we skip the details.
Taking limits and using the condition appearing in the hypotheses we obtain
$$
\mathcal{L}^{-1}(\frac{r}{f})(t) = \sum_{n=1}^\infty P_n(t)\,e^{\omega_n t} \; , 
$$
almost everywhere in $t$, in which $P_n(t)$ are polynomials of degree 
$\,r_n - 1$. Since we know
that $\mathcal{L}^{-1}(r/f) \in L^{p'}(\mathbb{R}_+)$, so is the above series, and 
we obtain (\ref{car-1}) as an equality in $L^{p'}(\mathbb{R}_+)$.
\end{proof}

\smallskip

\begin{rem}
The solution $(\ref{sol_lor})$ is not necessarily differentiable, see for instance
the  example appearing in {\rm \cite[p. 9]{GPR_CQG}}. This means, in particular, that
in complete generality we cannot even formulate initial value problem for equations 
of the form $f(\partial_t)\phi = J$.
\end{rem}

\section{The initial value problem}
In this section we discuss the initial value problem for equations of the form
\begin{equation} \label{lin_gen_1}
f(\partial_t) \phi = J \; , \quad \quad t \geq 0 \; ,
\end{equation}
in which $f$ is an (in principle arbitrary) analytic function, in the context of 
the theory developed in Section 2.

\subsection{Generalized initial conditions}

First of all we note that our abstract formula (\ref{sol_lor}) for the solution
$\phi$ to Equation (\ref{lin_gen_1}) tells us that --expanding the
analytic function $r$ appearing in (\ref{sol_lor}) as a power
series (or considering $r=r_d$ where $r_d$ is the series
(\ref{in_con111}))-- $\phi$ depends in principle on an infinite
number of arbitrary constants. However, this fact {\em does not}
mean that the equation itself is superfluous as sometimes argued in the
literature (see for instance \cite{EW}), as formula
(\ref{sol_lor}) for $\phi$ depends essentially on $f$ and $J$.
As remarked in \cite{GPR_CQG}, we may
think of $r$ as a ``generalized initial condition":

\begin{defi} \label{gics}
 A generalized initial condition for the equation
\begin{equation} \label{ivp0}
f(\partial_t) \phi = J
\end{equation}
is an analytic function $r_0$ such that $(\phi , r_0) \in D_f$
for some $\phi \in L^{p'}(0,\infty)$.
A generalized initial value problem is an equation such as $(\ref{ivp0})$ 
together with a generalized
initial condition $r_0$. A solution to a given generalized initial value problem
$\{ (\ref{ivp0}), r_0 \}$ is a function $\phi$ satisfying the conditions of
Definition $\ref{main_def}$ with $r = r_0$.
\end{defi}

Thus, given a generalized initial condition, we find a unique
solution for (\ref{lin_gen_1}) using (\ref{sol_lor}), much in the
same way as given {\em one} initial condition we find a unique
solution to  a first order linear ODE. As remarked 
after Corollary \ref{ic2}, there is no reason to believe that
(for a given $r$) the unique solution (\ref{sol_lor}) to
(\ref{ivp0}) will be analytic: within our general context, we can
only conclude that the solution belongs to the class $L^{p'}(0,\infty)\,$ for 
some $p' > 1$. It follows that classical initial value problems
do not make sense in full generality: Definition \ref{gics} is
what replaces them in the framework of our theory. Nonetheless, in
the next subsection we show that --provided $f$ and $J$ satisfy
some technical conditions-- we {\em can} define initial value problems
subject to a finite number of a priori local data.

\smallskip

\begin{example}
In \cite[Section 5]{B}, Barnaby considers D-brane
decay in a background de Sitter space-time. Up to some parameters,
the equation of motion is
\begin{equation} \label{neil}
e^{-2\, \square} (\square + 1) \phi = \alpha \phi^2 \; ,
\end{equation}
in which $\alpha$ is a constant. In de Sitter space-time we have
$\square = - \partial^2_t - \beta \partial_t$ for a constant
$\beta$, and so Equation (\ref{neil}) becomes
\begin{equation} \label{neil1}
e^{2(\partial_t^2 + \beta \partial_t)} (\partial_t^2 + \beta
\partial_t - 1) \phi = - \alpha \phi^2 \; .
\end{equation}

We check that $\phi_0(s) = 1/\alpha$ solves (\ref{neil1}) by using Equation
(\ref{aux1}). We set $f_0(s) = e^{2(s^2 + \beta s)}\,$ and we take as generalized
initial condition the series $r_0$ given by (\ref{in_con111}) with 
$d_j = \phi_0^{(j)}(0)$. We have,
\begin{eqnarray*}
\mathcal{L} ( f_0(\partial_t)\phi_0 )(s)  & = &
f_0(s) \mathcal{L}(\phi_0)(s) - \sum_{n=1}^\infty \sum_{j=1}^n \frac{f_0^{(n)}(0)}{n!}\, s^{n-j} \phi_0^{(j-1)}(0) \\
& = & f_0(s) \frac{1}{\alpha s} - \sum_{n=1}^\infty \frac{f_0^{(n)}(0)}{n!}\, s^{n-1} \frac{1}{\alpha} \\
& = & \frac{1}{\alpha s} \; .
\end{eqnarray*}
Thus, the left hand side of (\ref{neil1}) evaluated at $\phi_0$ equals
$-1/\alpha$, and the claim follows.

We linearize about $\phi_0$, recalling that rigorously
the domain of the nonlocal operator appearing in the left hand side of
Equation (\ref{neil1}) is formed by pairs $(\phi , r)$, see Definition
\ref{def0}. We obtain that if $\phi = \phi_ 0 + \tau\,\psi$ and 
$r = r_0 + \tau\,\tilde{r}$ for small $\tau$, then the deformation $\psi$ satisfies
the  equation
\begin{equation} \label{neil2}
e^{2(\partial_t^2 + \beta \partial_t)} (\partial_t^2 + \beta
\partial_t - 1) \psi + 2 \psi = 0 \; .
\end{equation}
We can solve this equation very easily for $\psi \in L^{p'}(0,\infty)$. Let us take
$p > 0$ such that $1/p + 1/p' =1$; we consider the entire function
\[
f(s) = e^{2(s^2 + \beta s)} (s^2 + \beta s - 1) + 2 \; ,
\]
and we let $r$ be a generalized initial condition, that is, $r$ is an
analytic function such that $r/f$ belongs to the class $H^p(\mathbb{C}_+)\,$. The
most general solution to (\ref{neil2}) is
\begin{equation}
\psi = \mathcal{L}^{-1}(r/f) \; .
\end{equation} 
\end{example}

\begin{rem}
\begin{itemize}

~

\item An equation similar to $(\ref{neil})$ appears in {\rm \cite{MN}}
in the context of Friedmann cosmology, and even more general equations of interest
for cosmology have been considered in {\rm \cite{Cal1}}. In the Friedmann case 
we have $\square = - \partial^2_t - 3 H(t) \partial_t$, in which $H$ represents 
the Hubble rate, and Equation $(\ref{neil1})$ must be changed accordingly. If we set
$f(t,s) = e^{2(s^2 + 3 H(t) s)} (s^2 + 3 H(t) s - 1) + 2$, we see that in this
case the ``symbol" $f(t,s)$ is beyond the class of symbols we are considering here.
We will come back to this example in another publication.

\item Besides the rigorous theory we considered in the previous papers
{\rm \cite{GPR_JMP,GPR_CQG}} and the one we develop herein, there exists a large 
body of literature on approximated, numerical, or formal solutions to nonlocal
equations, see for instance {\rm \cite{BK1,CN,Ko,Moeller,Vo}} and references therein. 
\end{itemize}
\end{rem}

\subsection{Classical initial value problems}

In this subsection we point out that
{\em if} we can unravel the abstract formula (\ref{sol_lor}) as in
Corollaries \ref{ic}, \ref{ic1} or \ref{ic2}, {\em then} we can define
and solve initial value problems depending on a finite number of initial  {\it local} data, and not on ``initial functions".

In Corollaries \ref{ic},
\ref{ic1}, and \ref{ic2} there are explicit formulas for
solutions to nonlocal equations of the form (\ref{lin_gen_1}). We would expect these 
formulas to help us in setting up initial
value problems. Now, Corollary \ref{ic} fixes completely the
solution using {\em only} $f$ and $J$, and therefore it leaves no
room for an initial value problem. On the other hand,
formula (\ref{car-1}) of Corollary \ref{ic2} depends on an
infinite number of parameters and, as explained in \cite{GPR_CQG},
this fact implies that we cannot ensure
differentiability of the solution (and hence existence of initial value problems) 
using {\em only} conditions on $f$ and
$J$. On the other hand, the explicit formula (\ref{car}) tells us that a
solution for the linear equation (\ref{lin_gen_1}) is uniquely
determined by $f$, $J$, and a {\em finite number} of parameters
related to the singularities of the quotient $r/f$. It is
therefore not unreasonable to expect that, by using these finitely
many parameters, we can set up consistent initial value problems,
as conjectured in \cite{Moeller}.

\smallskip

In order to ensure differentiability of solutions, we use the following 
two technical lemmas:

\begin{lem} \label{fo}
Let $J $ be a function such that $\mathcal{L}(J)$ exists and let $f$ be an 
analytic function. Suppose that there exist an integer $M \geq 0$ and a real number 
$\sigma >0$ such that
\begin{equation} \label{cond-1}
y \mapsto y^n\displaystyle\left( \frac{\mathcal{L}(J)(\sigma+iy)}{f(\sigma+iy)} \right)  \mbox{ belongs to }
L^1(\mathbb R)
\end{equation}
for each $n=0,\dots,M$; then the function
\begin{equation} \label{fo1}
t \mapsto \frac{1}{2 \pi i} \int_{\sigma - i \infty}^{\sigma +
i \infty} e^{s\,t}\, \left( \frac{\mathcal{L}(J)}{f} \right)(s)\,  ds
\end{equation}
is of class $C^M$.
\end{lem}
\begin{proof}
After a suitable change of variables we have
\[
\frac{1}{2 \pi i} \int_{\sigma - i \infty}^{\sigma +
i \infty} e^{s\,t}\, \left( \frac{\mathcal{L}(J)}{f} \right)(s)\, ds
=
\frac{e^{\sigma t}}{2 \pi} \int_{- \infty}^{\infty} e^{-i t y}  \left( \frac{\mathcal{L}(J)(\sigma+iy)}{f(\sigma+iy)} \right) dy \; .
\]
It follows that the left hand side is of class $C^M$ as a function of $t$,
since (\ref{cond-1}) implies we can differentiate the right hand side of the
above equation with respect to $t$ at least $M$ times.
\end{proof}

We remark that Doetsch's representation theorem (Theorem \ref{twDoe}) implies that 
it is enough to assume condition (\ref{cond-1}) for {\it some} $\sigma >0$ instead 
of {\it for every} $\sigma >0$. Lemma \ref{fo} implies the following:

\begin{lem} \label{tech}
Let the functions $f$ and $J$ satisfy the conditions of Corollary $\ref{ic1}$, 
and also that  they satisfy $(\ref{cond-1})$ for some $\sigma > 0$. Then, the 
solution $(\ref{car})$ to the nonlocal equation
\begin{equation} \label{main}
f(\partial_t)\phi(t) = J(t)
\end{equation}
is of class $C^M$ for $t \geq 0$, and it satisfies the identities
\begin{equation}  \label{car-40}
\phi^{(n)}(0) =  L_n + \sum_{i=1}^N\, \sum_{k=0}^{n} \left(\begin{array}{c} n \\
k \end{array} \right) \omega_i^k \left. \frac{d^{n-k}}{dt^{n-k}}
\right|_{t=0} P_i(t) \; , \quad n=0,\dots,M \; ,
\end{equation}
for some numbers $L_n\,$.
\end{lem}
%

This result is proven in \cite{GPR_JMP, GPR_CQG}. The numbers $L_n$, for 
$n=0,\dots,M$, are computed to be
\begin{equation} \label{ele}
L_n = \left.\frac{d^{n}}{dt^{n}}\right|_{t=0} \left( \frac{1}{2 \pi i}
\int_{\sigma - i \infty}^{\sigma + i\infty} e^{s t} \left(
\frac{{\cal L}( J )}{f} \right)(s) \, ds \right)\; .
\end{equation}

\smallskip

\smallskip

\smallskip

Our main result on initial value problems is the following. 
(The word ``generic" appearing
in the enunciate will be explained in the proof, see paragraph below Equation (33) ).

\smallskip

\begin{tw} \label{ivp2}
We fix real numbers $1 < p \leq 2$ and $p' > 0$ such that $1/p + 1/p' = 1$, and we also 
fix an integer $N \geq 0$.
Let $f$ be a function which is analytic in a region $D$ which contains
$\{ s \in \mathbb{C} : Re (s) > 0\}$, and let $J$ be a function in
$L^{p'}(\mathbb{R}_+)$ satisfying the conditions $\mathcal{L}(J) \in H^p(\mathbb{C}_+)$ 
and $\mathcal{L}(J)/f \in H^p(\mathbb{C}_+)$. We choose points $\omega_i$, $i=1, \cdots,N$, 
to the left of $Re(s)=0$, and positive integers $r_i\,$, $i =1, . . . ,N$. 
Set $K=\sum_{i=1}^N r_i$ and assume that for some $\sigma >0$ condition
$(\ref{cond-1})$ holds for each $n=0,\dots,M$, $M \geq K$.
Then, generically, given $K$ values $\phi_0, \dots, \phi_{K-1}$,
there exists a unique analytic function $r_0$ such that
\begin{itemize}
\item[$(\alpha)$] $\displaystyle \frac{r_0}{f} \in H^p(\mathbb{C}_+)$ and
it has a finite number of poles $\omega_i$ of order $r_i$, $i=1,\dots,N$;
\item[$(\beta)$] $\displaystyle \frac{\mathcal{L}(J) + r_0}{f} \in
        H^p(\mathbb{C}_+)$;
\item[$(\gamma)$] $\displaystyle \left| \frac{r_0}{f}(s) \right| \leq \frac{M}{|s|^q}$ for some $q \geq 1$ and $|s|$ sufficiently large.
\end{itemize}

Moreover, the unique solution $\phi = {\cal L}^{-1} \left( \frac{{\cal L}( J )+r_0}{f} \right)$ to Equation $(\ref{main})$ belongs to $L^{p'}(\mathbb{R}_+)$, 
is of class $C^K$, and it satisfies $\phi(0) = \phi_0, \dots, \phi^{(K-1)}(0)=\phi_{K-1}\,$.
\end{tw}

\begin{proof}
We consider the $K$ arbitrary numbers $\phi_n$, $n=0,1,\ldots ,K-1$. Recalling
(\ref{car-40}), Lemma \ref{fo}, and Lemma \ref{tech},
we set up the linear system
\begin{equation}  \label{car-4}
\phi_n =  L_n + \sum_{i=1}^N\, \sum_{k=0}^{n} \left(\begin{array}{c} n \\
k \end{array} \right) \omega_i^k \left. \frac{d^{n-k}}{dt^{n-k}}
\right|_{t=0} P_i(t) \; , \; \; \; n = 0 \dots K-1\; ,
\end{equation}
in which the unknowns are the coefficients of polynomials
$$
P_i(t) = a_{1,i} + a_{2,i}\, \frac{t}{1!} + \cdots + a_{r_i,i}\, \frac{t^{r_i - 1}}{(r_i -1)!} \; ,
$$
and the numbers $L_n$ are given by (\ref{ele}).

System (\ref{car-4}) can be solved
(generically, this is, away from the variety in $\mathbb{R}^N$ defined by all the points $\omega_i$
for which the main determinant of the linear system (33) vanishes) uniquely in terms of the data
$\phi_n\,$. We define $r_0$ on the half-plane
$\{ s \in \mathbb{C} : Re (s) > 0\}$ as follows:
\begin{equation} \label{l-1}
r_0(s) = f(s)\, {\mathcal L} \left( \sum_{i=1}^N P_i(t) e^{\omega_i t} \right)(s) \; .
\end{equation}
Then, on this half-plane we have the identity
\begin{equation} \label{l-2}
 {\mathcal L}^{-1}(r_0/f)(t) = \sum_{i=1}^N P_i(t) e^{\omega_i t} 
\end{equation}
for all $t \in \mathbb{R}_+$.  
Let us prove that $\dfrac{r_0}{f}$ belongs to $H^p(\mathbb{C}_+)$. 
First of all, we have the standard formula
$$\mathcal{L}(bt^n e^{\omega_j t})(s)= \frac{bn!}{(s-\omega_j)^{n+1}}$$ for  
any $n\in \mathbb{N}$ and $b \in \mathbb{C}$, and so the function $\dfrac{r_0}{f}$ 
is analytic on the half-plane $Re(s)>0$. We show that $\mu_{p}(\dfrac{r_0}{f},x)$ 
is uniformly bounded on $x>0$. 
For this, it is enough to prove that the function
$ \mu_{p}(\mathcal{L}(b t^n e^{\omega_j t}),x)$ is uniformly bounded for $x>0\,$. 
Indeed, since the poles $\omega_j:= a_j +ib_j$ of $\dfrac{r_0}{f}$ satisfy 
$ a_j= Re(\omega_j)<0$, we have $|x -a_j| >|a_j|>0$ for $x > 0$. 
Therefore:
\begin{eqnarray}
\int_{-\infty}^{-\infty} \left| \mathcal{L}(b t^n e^{\omega_j t})(x +i y))  \right|^p dy & = & |b|n! \int_{-\infty}^{-\infty}\frac{dy}{|x +i y-\omega_j|^{(n+1)p}} \nonumber \\
& = & |b|n! \int_{-\infty}^{\infty}\frac{dy}{((x-a_j)^2+(y-b_j)^2)^{\frac{n+1}{2}p}} \nonumber \\
& = & |b|n! \int_{-\infty}^{\infty}\frac{d \xi}{((x-a_j)^2+\xi^2)^{\frac{n+1}{2}p}} \nonumber \\
& < & 2|b|n! \int_{0}^{\infty}\frac{d \xi}{(a_j^2+\xi^2)^{\frac{n+1}{2}p}} \label{p1case}\\
& < & \infty. \nonumber
\end{eqnarray}
Thus, $ \mu_{p}(\mathcal{L}(b t^n e^{\omega_j t}),x)$ is uniformly bounded for $x>0$. 
This shows that $\dfrac{r_0}{f}$ belongs to $H^p(\mathbb{C}_+)$, and ($\alpha$) 
follows. 

We easily check that $r_0/f$ also satisfies conditions ($\beta$), and ($\gamma$)
appearing in the enunciate of the theorem; we omit the details. 

Now we define the function
\begin{equation} \label{car-7}
\phi(t) = \frac{1}{2 \pi i}
\int_{\sigma - i \infty}^{\sigma + i\infty} e^{s t} \left( \frac{{\cal L}( J )}{f} \right)(s) \, ds 
+ \sum_{i=1}^N P_i(t)\,e^{\omega_i\, t} 
\end{equation}
in $L^{p'}(\mathbb{R}_+)$. We claim that this function solves Equation (\ref{main}) 
and that it satisfies the conditions appearing in the enunciate of the theorem.
In fact, the foregoing analysis implies that
\[
 \phi(t) = \frac{1}{2 \pi i} \int_{\sigma - i \infty}^{\sigma +
i\infty} e^{s t} \left( \frac{{\cal L}( J )}{f} \right)(s) \, ds +
{\mathcal L}^{-1}(r_0/f)(t) \; ,
\]
and this is precisely the unique solution to (\ref{main}) appearing in Corollary 
\ref{ic1} for $r=r_0$.

\smallskip

It remains to show that this solution satisfies $\phi^{(n)}(0) = \phi_n$ for
$n = 0 , \dots, K-1$.
Indeed, condition (\ref{cond-1}) tells us that $\phi(t)$ is at least of class 
$C^K$ and clearly
\begin{equation}  \label{car-9}
\phi^{(n)}(0) =  L_n + \sum_{i=1}^N\, \sum_{k=0}^{n} \left(\begin{array}{c} n \\
k \end{array} \right) \omega_i^k \left. \frac{d^{n-k}}{dt^{n-k}}
\right|_{t=0} P_i(t)
\end{equation}
in which
\begin{equation} \label{ele-9}
L_n = \left.\frac{d^{n}}{dt^{n}}\right|_{t=0} \left( \frac{1}{2 \pi i}
\int_{\sigma - i \infty}^{\sigma + i\infty} e^{s t} \left(
\frac{{\cal L}( J )}{f} \right)(s) \, ds \right)\; .
\end{equation}
Comparing (\ref{car-4}) and (\ref{car-9}) we obtain $\phi^{(n)}(0) = \phi_n$, 
$n = 0, \dots, K-1$.
\end{proof}

\smallskip

\begin{rem}

~

\begin{itemize}
\item The proof of Theorem $\ref{ivp2}$ breaks down if $p=1$, as we lose the 
uniform bound $(\ref{p1case})$. This is the reason why, as advanced in Subsection 
$2.1$, ``$L^1$-correspondence theorems" (see {\rm \cite[Theorem 2]{Chu}} and 
{\rm \cite[Theorem 28.2]{Doetsch}}) do not allow us to 
obtain a meaningful theory of initial value problems. On the other hand, it is 
not hard to develop Section $2$ within the framework of {\rm \cite{Chu,Doetsch}}.

\item Theorem $\ref{ivp2}$ tells us that we can freely chose the first 
$K$ derivatives $\phi^{(n)}(0)$, $n=0, \dots , K-1$ of the solution $\phi$ to
Equation $(\ref{main})$ but, from $n = K$ onward, if the derivative
$\phi^{(n)}(0)$ exists, it is completely determined by $(\ref{car-9})$ and 
$(\ref{ele-9})$. 
Thus, it does not make sense to formulate an initial value problem 
$f(\partial_t)\phi(t) = J(t)$, $\phi^{(n)}(0) = \phi_n$, with more than
$K$ arbitrary initial conditions.
\end{itemize}
\end{rem}

The above proof shows that the {\em a priori} given points $\omega_i$
become the poles of the quotient $r_0/f$, that the {\em a priori} given numbers 
$r_i$ are their respective orders, and that {\em it is essential}  to give
this information in order to have meaningful initial value problems. If
no points $\omega_i$ are present, the solution to the nonlocal equation 
(\ref{main}) is simply
\[
\phi = \mathcal{L}^{-1}(\mathcal{L}(J)/f)\; ,
\]
a formula which fixes completely (\,for $f$ and $J$ satisfying (\ref{cond-1})\,) 
the values of the 
derivatives of $\phi$ at zero. This discussion motivates the following definition:

\begin{defi} \label{civp}
A classical initial value problem for nonlocal equations is a triplet
formed by a nonlocal equation
\begin{equation} \label{ivp00}
f(\partial_t) \phi = J \; ,
\end{equation}
a finite set of data:
\begin{equation} \label{total_data}
\left\{ N \geq 0\; ; \; \; \{\omega_i \in \mathbb{C} \}_{1 \leq i \leq N}\; ;
\; \; \{r_i \in \mathbb{Z}, r_i > 0 \}_{1 \leq i \leq N}\; ;
\; \; \{\{ \phi_n \}_{0 \leq n \leq K-1}, \textstyle{K = \sum_{i=1}^N r_i} \} \right\} \; ,
\end{equation}
and the conditions
\begin{equation} \label{defic}
\phi(0) = \phi_0\; , \quad \phi'(0)= \phi_1\; , \quad \cdots \; .
\phi^{(K-1)}(0)=\phi_{K-1} \; .
\end{equation} 
A solution to a classical initial value problem given by
$(\ref{ivp00})$, $(\ref{total_data})$ and $(\ref{defic})$ is a pair
$(\phi,r_0) \in D_f$ satisfying the conditions of Definition $\ref{main_def}$ 
with $r = r_0$ such that $\phi$ is differentiable at zero and $(\ref{defic})$ holds.
\end{defi}

Theorem \ref{ivp2} implies that this definition makes sense.
From its proof we also deduce the following important corollaries:

\begin{cor}
We fix $1 < p \leq 2$ and $p' > 0$ such that $1/p + 1/p' = 1$.
Let $f$ be function which is analytic on the half-plane 
$\{ s \in \mathbb{C} : Re(s) > 0 \}$, 
and fix a function $J$ in $L^{p'}(0,\infty)$
such that $\mathcal{L}(J) \in H^p(\mathbb{C}_+)$ and 
$\mathcal{L}(J)/f \in H^p(\mathbb{C}_+)$.
Then, generically (in the sense of Theorem $\ref{ivp2}$), the classical initial 
value problem $(\ref{ivp00})$--$(\ref{defic})$ has a unique solution which 
depends smoothly on the initial conditions $(\ref{defic})$.
\end{cor}
\begin{proof}
The function (\ref{car-7}) is the unique solution to (\ref{ivp00}) and it
satisfies (\ref{defic}). Moreover, the coefficients of the polynomials $P_i$ 
appearing in (\ref{car-7}) depend smoothly on the data $\phi_n$ and $\omega_i$, 
since they are solutions to the linear problem (\ref{car-4}).
\end{proof}

\begin{cor} \label{cor20}
We fix $1 < p \leq 2$ and $p' > 0$ such that $1/p + 1/p' = 1$.
Let $f$ be function which is analytic on the half-plane 
$\{ s \in \mathbb{C} : Re(s) > 0 \}$, 
and fix $J$ in $L^{p'}(0,\infty)$
such that $\mathcal{L}(J) \in H^p(\mathbb{C}_+)$ and 
$\mathcal{L}(J)/f \in H^p(\mathbb{C}_+)$. 
Suppose also that there exist a positive integer $K$ such that for some $\sigma >0$
condition $(\ref{cond-1})$ holds, and consider a set of complex numbers
$ \phi_0,\phi_1,\cdots, \phi_{K-1}$. Then, there exists a classical initial 
value problem $(\ref{ivp00})$--$(\ref{defic})$.
 \end{cor}
\begin{proof}
Since $K>0$ we can find positive integers $N, r_1, \cdots , r_N$ such that 
$K= \sum_{i=1}^{N}r_i$. Also, 
we can choose $N$ complex numbers $\omega_i, 1\leq i \leq N$ to the left of 
$Re(s)= 0$.
Thus, we have constructed the data (\ref{total_data}). The nonlocal equation 
(\ref{ivp00}) is $f(\partial_t)\phi(t) = J(t)$, and conditions (\ref{defic})
are determined by the given complex numbers $\phi_0, \cdots \phi_{K-1}$.
\end{proof}

\begin{rem}
We note that there is at least one natural way to choose the poles
$\omega_i\,$ appearing in the proof of Corollary $\ref{cor20}$: 
if the symbol $f$ has $K$ zeroes (counting with multiplicities), say
$\{z_1,z_2\cdots,z_K\}$, to the left of $Re(s)<0$, and condition 
$(\ref{cond-1})$ holds 
for $n = 0 , \cdots , K$, we can obtain a unique solution
to the initial value problem $(\ref{ivp00})$--$(\ref{defic})$
if we choose $\omega_i = z_i$. In this case, the determinant of the
linear system $(\ref{car-4})$ is precisely the non-zero determinant of the
$K\times K$-Vandermonde matrix $A_K=(a_{ji})$ with $a_{ji} =z_{i}^{j-1}$, 
$1\leq j,i\leq K$.
\end{rem}

\section{Zeta nonlocal equations}

In this section we apply our previous approach to zeta nonlocal equations of the form
\begin{equation} \label{eqzeta_01}
\zeta_h(\partial_t)\phi(t)=J(t) \; , \; \; \; \; \; t\geq 0 \; ,
\end{equation}
in which $h$ is a real parameter and the symbol $\zeta_h$ is the shifted Riemann 
zeta function
$$\zeta_h(s) := \zeta(h+s)=\sum_{n=1}^{\infty}\dfrac{1}{n^{s+h}}\; .$$
These equations are motivated by the cosmological models appearing in  
\cite{AV2,D,D1,D2}. In turn, it has been called to the authors' attention by 
A. Koshelev that
the approach of \cite{AV2} is at least partially based on \cite{Ko}, in which the
author introduces a formal schemme for analyzing some nonlocal equations of interest
for cosmology.

\smallskip

We recall some properties of the Riemann zeta function 
$\zeta(s):=\sum_{n=1}^{\infty}\dfrac{1}{n^s}$, $Re(s)>1$, following \cite{KaVo}. 
It is analytic on its domain of definition and it has an analytic continuation to the whole 
complex plane with the exception of the point $s=1$, 
at which it has a simple pole with residue $1$.  The analytic continuation  of the Riemann 
zeta function will be also 
denoted by $\zeta$, and we will refer to it also as the Riemann zeta function.

From standard properties of the Riemann zeta function (see \cite{KaVo}) we have that the shifted 
Riemann zeta function $\zeta_h$ is analytic for $Re(s)>1-h$, and  uniformly and absolutely
convergent for $Re(s)\geq \sigma_0 >1-h$. We also 
find (see \cite[Chp. I.6]{KaVo}) that $\zeta_h$ has infinite zeroes at the points 
$\{-2n-h : n\in \mathbb{N}\}$ (we call them ``trivial zeroes") and that it also has 
``nontrivial" zeroes in the region $-h <  Re(s) < 1-h$ (we call this region 
the ``critical region" of $\zeta_h$).

\smallskip

\smallskip

The Euler product expansion for the shifted Riemann zeta function is
$$\zeta_h(s)=\prod_{p\in\mathcal{P}}\left(1-\dfrac{1}{p^{s+h}}\right)^{-1}\; ,$$
where $\mathcal{P}$ is the set of the prime numbers. Therefore, for $Re(s)=\sigma >1-h$, we have
\begin{eqnarray}\label{eqzeta_02}
\left| \dfrac{1}{\zeta_h(s)} \right| & = & \left|\prod_{p\in \mathcal{P}}\left(1-\dfrac{1}{p^{s+h}}\right) \right|
= \left|\sum_{n=1}^{\infty}\dfrac{\mu(n)}{n^{s+h}} \right|\nonumber \\
&\leq& \sum_{n=1}^{\infty}\dfrac{1}{n^{\sigma+h}}\leq 1+\int_1^{\infty}\dfrac{dx}{x^{\sigma+h}}=\dfrac{\sigma+h}{\sigma+h-1}\; ,
\end{eqnarray}
where $\mu(\cdot)$ is the M\"oebius function defined as follows: $\mu(1)=1,\; \; \mu(n)=0$ if $n$ is 
divisible by the square of a prime, and $\mu(n)=(-1)^k$ if n is the product of $k$ distinct prime 
numbers, see \cite[Chp. II.2]{KaVo}.

\smallskip

We study Equation (\ref{eqzeta_01}) for values of $h$ in the region $(1,\infty)$, since in this case 
$\zeta_h$ is analytic for $Re(s)>0$ and the theory developed in the previous sections apply.  
We start with the following lemma:

\begin{lem} \label{a2}
We fix $1 < p \leq 2$ and $p' > 0$ such that $1/p + 1/p' = 1$. 
Let us assume that $J\in L^{p'}(\mathbb{R}_+)$ and that
$\mathcal{L}(J)$ is in the space $H^p(\mathbb{C}_+)$. Then, the function 
$F =\dfrac{\mathcal{L}(J)}{\zeta_h}$ belongs to $H^p(\mathbb{C}_+)\,$.
\end{lem}
\begin{proof}  We have:

\begin{itemize}
\item[1.] The function $F$ is clearly analytic for every $s$ such that
$Re(s)>0$.
\item[2.] Since $\mathcal{L}(J) \in H^p(\mathbb{C}_+) $ we have that $\mu_p(\mathcal{L}(J),x)$ is uniformly bounded for $x>0$,
Now for $x>0$ and using inequality (\ref{eqzeta_02}), we obtain
\begin{eqnarray*}
\mu_p(F,x)&=&\left( \int_{\mathbb{R}}\left|\dfrac{1}{\zeta_h(x+iy)}\mathcal{L}(J)(\sigma+iy)\right|^p  dy \right)^{\frac{1}{p}} \\
&\leq& \dfrac{x+h}{x+h-1}\left( \int_{\mathbb{R}}|\mathcal{L}(J)(x+iy)|^p dy\right)^{\frac{1}{p}} <\infty \; .
\end{eqnarray*}
Since the function $x\to \dfrac{x+h}{x+h-1}$ is uniformly bounded for $x\geq 0$, the result follows.
%
%
%
%
\end{itemize}
\end{proof}

Now we note that we can replace the general condition (\ref{cond-1}) for the following 
assumption on the function $J$:
\begin{itemize}
\item[{\bf (H)}] For some $M \geq0$, for each $ n=1,2,3, \cdots ,M$ and for some $\sigma >0$ 
we have, $$y\to y^n\mathcal{L}(J)(\sigma+iy)\in L^1(\mathbb{R}).$$
\end{itemize}

\noindent Condition $\bf{(H)}$ is enough to ensure differentiability of the function
 $$t\to \dfrac{1}{2\pi i}\int_{\sigma-i\infty}^{\sigma+i\infty}e^{st} \left(\dfrac{\mathcal{L}(J)}{\zeta_h}\right)(s)ds \; ,$$
as asked in the hypotheses of Lemma \ref{fo}. Our main theorem on classical initial value problems for the zeta non-local equation (\ref{eqzeta_01})
%
is the following
\begin{tw} \label{ZetaIVP01}
We fix $1 < p \leq 2$ and $p' > 0$ such that $1/p + 1/p' = 1$. 
Let $\zeta_h$ be the shifted Riemann zeta function, and assume that $J \in L^{p'}(\mathbb{R}_+)$ 
with $\mathcal{L}(J) \in H^p(\mathbb{C}_+)$. We also fix a number $N \geq 0$, a finite number
of points $\omega_i$, $i =1, . . . ,N$, to the left of $Re(s)=0$, and a finite number of positive
integers $r_i$. We set $K=\sum_{i=1}^N r_i$ and we assume that condition $\bf{(H)}$ holds for all
$n=0,\dots,M$, $M \geq K$. Then, generically, given $K$ initial conditions, 
$\phi_0, \dots, \phi_{K-1}$, there exists a unique analytic function $r_0$ such that
\begin{itemize}
\item[$(\alpha)$] $\displaystyle \dfrac{r_0}{\zeta_h} \in H^p(\mathbb{C}_+)$ and it has a finite number of poles $\omega_i$ of order $r_i$, $i=1,\dots,N$ to the left of $Re(s)=0\,$;
\item[$(\beta)$] $\displaystyle \frac{\mathcal{L}(J)+r_0}{\zeta_h} \in H^p(\mathbb{C}_+)$;
\item[$(\gamma)$] $\displaystyle  \left| \frac{r_0}{\zeta_h}(s) \right| \leq \frac{M}{|s|^q}$ for some $q \geq 1$ and $|s|$ sufficiently large.
\end{itemize}
Moreover, the unique solution $\phi$ to Equation $(\ref{eqzeta_01})$ given by
$(\ref{sol_lor})$ with $r = r_0$ is of class $C^K$ and it satisfies
$\phi(0) = \phi_0, \dots, \phi^{(K-1)}(0)=\phi_{K-1}$.
\end{tw}
\begin{proof}
The proof consists in checking that the hypotheses of Theorem \ref{ivp2}
hold. In fact, $\zeta_h$ is analytic on $\{ s \in \mathbb{C} : Re(s) > 0 \}$, and Lemma \ref{a2} tells us that $\mathcal{L}(J)/\zeta_h$ belongs to $H^p(\mathbb{C}_+)$.
\end{proof}

\section{An $L^2(\mathbb{R}_+)$-theory for linear nonlocal equations}

In this section we show that in the $p=p'=2$ case of the foregoing theory, we can justify 
rigorously the interpretation of $f(\partial_t)$ as an operator in infinitely many derivatives 
on an appropriated domain. Our approach uses a notion of  analytic vectors, motivated by 
Nelson's classical paper \cite{Nelson}.

\begin{defi}\label{defEntire}
Let $A$ be a linear operator from a Banach space $\mathbb{V}$ to itself, and let 
$f:\mathbb{R} \to \mathbb{C}$ be a complex valued function, such that $f^{(n)}(0)$ 
exist for all $n\geq 0$. We say that $v \in \mathbb{V}$ is a $f$-analytic vector for
$A$ if $v$ is in the domain of $A^n$ for all $n\geq 0$ and the series
$$\sum_{n=0}\dfrac{f^{(n)}(0)}{n!}A^n v \; ,$$
defines a vector in $\mathbb{V}$.
\end{defi}

As stated in Section 2, the Paley-Wiener theorem (see \cite{Har,KoYo}) is the following 
special case of Doetsch's representation theorem:

\begin{tw} \label{pw}
The following assertions hold:
\begin{itemize}
\item[1)] If $g\in L^2(\mathbb{R}_+)$, then $\mathcal{L}(g)\in H^2(\mathbb{C}_+)$.
\item[2)] Let $G\in H^2(\mathbb{C}_+)$. Then the function 
$$g(t) = \frac{1}{2 \pi i}\int_{\sigma-i \infty}^{\sigma + i \infty} e^{s t} G(s) ds \; , \sigma > 0 \; , $$
is independent on $\sigma$, it belongs to $L^2(\mathbb{R}_+)$ and it satisfies 
$G=\mathcal{L}(g)$.
\end{itemize}
Moreover the Laplace transform $\mathcal{L}: L^2(\mathbb{R}_+) \to H^2(\mathbb{C}_+)$ 
is a unitary operator.
\end{tw}

We examine existence of analytic vectors for the operator $\partial_t$ 
on $L^2(\mathbb{R}_+)$.

\begin{lem}\label{Avec}
Let $f$ be an analytic function on a region containing zero, and
let $R_1$ be the maximun radius of convergence of the Taylor series
$f_T(s) := \sum_{n=0}^\infty \frac{f^{(n)}(0)}{n!}\, s^n \; $.
 \begin{itemize}
\item [a)] If $p$ is a polynomial on $\mathbb{R}_+$, and $I$ a finite interval on 
$\mathbb{R}_+$, then the function $\psi:= p\cdot \chi_I$ is
 an $f$-analytic vector for $\partial_t$ on $L^2(\mathbb{R}_+)$.
 \item [b)] Let $R_1>1$.
 If $\psi \in C^{\infty}(\mathbb{R}_+)\cap L^2(\mathbb{R}_+)$ such that for all 
 $n\geq 1$ and some
 $h\in L^2(\mathcal{I})$ for $\mathcal{I}$ equal to either $\mathbb{R}$ or 
 $\mathbb{R}_+$, 
 we have $||\psi^{(n)}||_{L^2(\mathbb{R}_+)}\leq c(n)||h||_{L^2(\mathcal{I})}$, with 
 $\{c(n)\}_{n\in \mathbb{N}}=:c \in l^1(\mathbb{N})$. Then $\psi$ is an 
 $f$-analytic vector for 
 $\partial_t$ on $L^2(\mathbb{R}_+)$.
 \end{itemize}
\end{lem}
\begin{proof}
Part $a)$ is immediate. For $b)$ we have: Since $R_1>1$, the sequence 
$\left\{\frac{f^{(n)}(0)}{n!}\right\}_{n\in \mathbb{N}}$ 
 is bounded; therefore, there is a positive constant $C$ such that 
 $|\frac{f^{(n)}(0)}{n!}|\leq C$ for every $n \in \mathbb{N}$. Now, by 
 Minkowski inequality we have,
\begin{eqnarray*}
||\sum_{n=0}^\infty \frac{f^{(n)}(0)}{n!} \psi^{(n)}||_{L^2(\mathbb{R}_+)}& \leq & \sum_{n=0}^\infty ||\frac{f^{(n)}(0)}{n!} \psi^{(n)}||_{L^2(\mathbb{R}_+)}\\
&\leq & C \sum_{n=0}^\infty ||\psi^{(n)}||_{L^2(\mathbb{R}_+)}\\
&\leq & C ||h||_{L^2(\mathcal{A})} \sum_{n=0}^\infty c(n)\\
&\leq & C ||c||_{l^1(\mathbb{N})} ||h||_{L^2(\mathcal{A})} \; .
\end{eqnarray*}
\end{proof}


%
\begin{rem}
We present two large families of functions $\psi$ which meet
conditions $b)$ of Lemma $\ref{Avec}$.
\begin{itemize}
\item Let $k>1$ be a parameter and consider the functions 
$\psi(t):=e^{-\frac{t}{k}}, t\geq 0$; then we have 
$\psi^{(n)}(t)=\frac{(-1)^n}{k^n}e^{-\frac{t}{k}}$, therefore 
$|\psi^{(n)}(t)|=\frac{1}{k^n}e^{-\frac{t}{k}}$ and if we define 
$c(n):=\frac{1}{k^n}$, we have $\{c(n)\}_{n \in \mathbb{N}}\in l^1(\mathbb{N})$.
\item We recall that an arbitrary entire function $\phi$ of exponential
type $\tau$ which is also in $L^2(\mathbb{R})$, satisfies the generalized 
$L^2$-Berstein inequality
$$||\phi^{(n)}||_{L^2(\mathbb{R})}\leq \tau^n ||\phi||_{L^2(\mathbb{R})}\; ,$$ see 
$\cite[Chp. 3]{Ni}$. We denote by $Exp_1^2(\mathbb{C})$ the space of entire 
functions of exponential type $\tau<1$ which are $L^2$-functions on $\mathbb{R}$.
Then, for any $\phi \in Exp_1^2(\mathbb{C})$ with exponential type $\tau_\phi > 0$,
the function $\psi := \chi_{\mathbb{R}_+} \cdot \phi$ satisfies part $b)$ of the
lemma with $h = \phi$. In particular, let $\phi$ be a smooth function on 
$\mathbb{R}$ with compact support in $[-\tau,\tau]\subset \mathbb{R}$ with 
$\tau <1$. Then its Fourier transform
$$\mathcal{F}(\phi)(t):=(2\pi)^{-1/2}\int_{-\tau}^{\tau} e^{-ixt}\phi(x)dx$$
is an $L^2(\mathbb{R})$-function and it has an extension to an entire function 
$\Phi$ which is of exponential type 
$\tau$. The function $\psi:= \chi_{\mathbb{R}_+} \cdot \Phi$ satisfies part $b)$ 
of the lemma.
\end{itemize}
\end{rem}

Definition \ref{def0} of the operator $f(\partial_t)$ restricts to the present 
$L^2$-context. We state it explicitly for the reader's convenience.

\begin{defi} \label{def00}
Let $f$ be an analytic function on a region which contains zero and the half-plane $\{ s \in \mathbb{C} : Re(s) > 0 \}$, 
and let $\mathcal{H}$ be the space of all functions which are analytic on regions of $\mathbb{C}$.
We consider the subspace $D_f$ of $L^2(\mathbb{R}_+) \times \mathcal{H}$
consisting of all the pairs $(\phi , r)$ such that
\begin{equation} \label{tr00}
\widehat{(\phi , r)} = f\,{\mathcal L}(\phi) - r
\end{equation}
belongs to the space $H^2(\mathbb{C}_+)$.
The domain of $f(\partial_t)$ as a linear operator from $L^2(\mathbb{R}_+) \times \mathcal{H}$
to $L^2(\mathbb{R}_+)$ is the set $D_f\,$. If $(\phi , r) \in D_f$ then
\begin{equation} \label{tr10}
f(\partial_t)\,(\phi , r) = {\mathcal L}^{- 1} (\,\,\widehat{(\phi,r)} \,\,) =
{\mathcal L}^{- 1} ( f\,{\mathcal L}(\phi)
- r ) \; .
\end{equation}
\end{defi}

We show that this definition is not empty, and that in fact the domain
$D_f$ is quite large. The fact that the Laplace transform is an unitary operator
plays an essential role at this point:

\begin{prop} \label{ext0}
Let $f$ be a function which is analytic on a region containing $\mathbb{C}_+$,
and let $R_1 > 1$ be the maximum radius of convergence of the Taylor series
$f_T(s) := \sum_{n=0}^\infty \frac{f^{(n)}(0)}{n!}\, s^n$.
Let $\phi$ be a smooth $f$-analytic vector for $\partial_t$ in
$L^2(\mathbb{R}_+)$ and suppose that the sequence $\{ d_{j} = \phi^{(j)}(0) \}$ 
satisfies the condition of Lemma $\ref{r_conv}$. Then, there exists an analytic function $r_e$ on
$\mathbb{C}_+$ such that $(\phi , r_e)$ is in the domain $D_f$ of $f(\partial_t)$.
\end{prop}
\begin{proof}
From equation (\ref{in_con111}) and Lemma \ref{r_conv} we can define the following 
analytic function $\hat{\phi}$ on $\{ z \in \mathbb{C} : |z | < R_1 \}\,$:
\begin{eqnarray*}
\hat{\phi}(s)&:=& \sum_{n=0}^\infty \frac{f^{(n)}(0)}{n!}s^n \mathcal{L}(\phi)(s)-\sum_{n=1}^\infty \sum_{j=1}^n \frac{f^{(n)}(0)}{n!}\,
\phi^{(j-1)}(0) \, s^{n-j}\\
&=&\sum_{n=0}^\infty \frac{f^{(n)}(0)}{n!}\left( s^n \mathcal{L}(\phi)(s)-\sum_{j=1}^n \phi^{(j-1)}(0) \, s^{n-j}\right)\\
&=&\sum_{n=0}^\infty \frac{f^{(n)}(0)}{n!}\mathcal{L}(\partial_t^n(\phi))(s)\\
&=&\sum_{n=0}^\infty \frac{f^{(n)}(0)}{n!}\int_{0}^{\infty}e^{-st}\partial_t^n(\phi)(t)dt\; .
\end{eqnarray*}
Using this last equality and the fact that the Laplace transform is an isometric isomorphism from $L^2(\mathbb{R}_+)$ onto $H^2(\mathbb{C}_+)$, see Theorem \ref{pw},
we have
\begin{equation}\label{eq_exta}
\hat{\phi}(s)=\int_{0}^{\infty}e^{-st}\sum_{n=1}^\infty \frac{f^{(n)}(0)}{n!}\partial_t^n(\phi)(t)dt \; .
\end{equation} 
Therefore, on the disk $\{ z \in \mathbb{C} : |z | < R_1 \}\,$ we have the equation
\begin{equation} \label{ext}
\mathcal{L} \left( \sum_{n=1}^\infty \frac{f^{(n)}(0)}{n!}\partial_t^n(\phi)(t) \right)(s) = 
f(s) \mathcal{L}(\phi)(s) - r(s) \; ,
\end{equation}
where $r(s) = \sum_{n=1}^\infty \sum_{j=1}^n \frac{f^{(n)}(0)}{n!}\,
\phi^{(j-1)}(0) \, s^{n-j}$.

\smallskip

Now, we stress the fact that the right hand side of (\ref{eq_exta}) belongs to 
the Hardy space $H^2(\mathbb{C}_+)$; it follows that we can extend $\hat{\phi}(s)$ 
via analytic continuation to a function 
$\hat{\phi}_e$ on the half-plane $Re(s)>0$. This function belongs to the Hardy 
space $H^2(\mathbb{C}_+)$ by construction. Also, we note that
Equation (\ref{ext}) implies
\begin{equation}\label{eq_extar}
r(s)=f(s)\mathcal{L}(\phi)(s)-\hat{\phi}_e(s)
\end{equation}
for $|s|<R_1$. However, the right hand side of (\ref{eq_extar}) is
defined on $Re(s) > 0$, and therefore it defines
an analytic continuation $r_e$ of the series $r$ to the half-plane $Re(s)>0$. 
Thus, on $Re(s)>0$ we have the equation:
$$\hat{\phi}_e(s)= f(s)\mathcal{L}(\phi)(s)-r_e(s)\; .$$
Since $\hat{\phi}_e$ is in $H^2(\mathbb{C}_+)$, we have that $(\phi , r_e) \in D_f$.
\end{proof}

\begin{example}\label{exa_riemann}
It follows from our discussion on the function $\zeta_h$ introduced in Section 4,
that $\zeta_h$ is analytic around zero for $h > 1$ large enough, and that therefore
we have the power series expansion
$$\zeta_h(s)=\sum_{n=0}^{\infty}a_n(h)s^n \; .$$ Again for appropriate
$h >1$, we can assume that its maximum radius of convergence is $R_1 > 1$.
Then, Lemma \ref{Avec} implies that there exists a large class of
$\zeta_h$-analytic vectors $\phi \in L^2(\mathbb{R}_+)$ and moreover,
the above proposition applies. Thus, $\zeta_h(\partial_t)$
is a well-defined operator on (a subspace of) $L^2(\mathbb{R}_+)\times \mathcal{H}$.
\end{example}

\smallskip

An easy corollary of Proposition \ref{ext0} is the following:

\begin{cor}\label{cor_fin}
Let $f$ be an entire function and let $\phi$ be a smooth $f$-analytic vector for 
$\partial_t$ in $L^2(\mathbb{R}_+)$.
Suppose that the sequence $\{ d_{j} = \phi^{(j)}(0) \}$ satisfies
$$d_j\leq CR^j$$ for $0<R<1$. Then $(\phi,r)\in D_f$, in which $r$ is the series 
defined in the above proof.
\end{cor}

The proof of Corollary \ref{cor_fin} consists in noting that the stated 
hypotheses allow us to apply Lemma 2.3.

\smallskip

Proposition \ref{ext0} and Corollary \ref{cor_fin} imply
that (if $f$ is entire) the operator $f(\partial_t)$ is an operator in infinitely many 
derivatives on the space of smooth $f$-analytic vectors. In fact, let $\phi$ be a 
$f$-analytic vector for $\partial_t$ in $L^2(\mathbb{R}_+)$ and let
$$r(s) = \sum_{n=1}^\infty \sum_{j=1}^n \frac{f^{(n)}(0)}{n!}\,
\phi^{(j-1)}(0) \, s^{n-j} \; .$$
If conditions of corollary \ref{cor_fin} hold, then
$$\mathcal{L}\left( \sum_{n=0}^\infty \frac{f^{(n)}(0)}{n!} \partial_t^n(\phi) \right)
= f \mathcal{L}(\phi) - r = \mathcal{L}( f(\partial_t)\phi ) \; ,$$
and therefore
$$ f(\partial_t)\phi=\sum_{n=0}^\infty \frac{f^{(n)}(0)}{n!} \partial_t^n(\phi)\; .$$

%


\section{Discussion}

As mentioned in Section 1, the following nonlocal equation appears naturally in the study of a zeta nonlocal scalar field model in string theory (See \cite{D,D1,D2}; recall that we are using signature so that in the $1+0$ dimensional case the d'Alambert operator is $\partial_t^2$):
\begin{equation}\label{Zeq_04}
\zeta(\partial_t^2+h)\phi=\mathcal{AC}\sum_{n=0}^{\infty}n^{-h}\phi^n\; , \quad
t\geq 0 \; ,
\end{equation}
where $\mathcal{AC}$ means analytic continuation. We stress, after \cite{AV},  that it is natural to consider the restriction $t \geq 0$ since classical versions of cosmological models contain singularities at the beginning of time.

Equation (\ref{Zeq_04}) motivates the study of the following nonlocal linear equations
\begin{equation} \label{Zeq_044}
\zeta(\partial_t^2+h) \phi = J \; , \quad t \geq 0 \; ,
\end{equation}
for appropriate functions $J$. Interestingly, the behavior of the symbol $\zeta(s^2+h)$ is quite different 
to the behavior of the symbol $\zeta(s+h)$ appearing in Section 4.
We show here that a study of Equation (\ref{Zeq_044}) requires
a generalization of the theory developed in the above sections.

First of all,  from the properties of the Riemann zeta function, we observe that the symbol
\begin{equation} \label{nszf}
\zeta(s^2+h)=\sum_{n=0}^{\infty}\dfrac{1}{n^{s^2+h}}
\end{equation}
is analytic in the region $\Gamma:= \{ s \in \mathbb{C} : Re(s)^2-Im(s)^2>1-h \}$, which is not a half-plane; also we can note that its poles are the vertices of the hyperbolas $Re(s)^2-Im(s)^2=1-h $ and its critical region is the set $\{ s \in \mathbb{C} : -h< Re(s)^2-Im(s)^2<1-h \}$.
In fact, according to the value of $h$ we have:
\begin{itemize}
\item [i)]For $h>1$, $\Gamma$ is the region limited by the interior of the dark hyperbola $Re(s)^2-Im(s)^2=1-h $ containing the real axis:
\begin{center} 
\includegraphics[width=4cm, height=4cm]{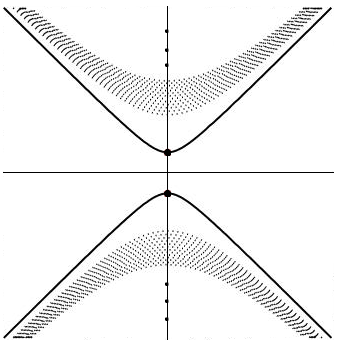}
\end{center}
{\footnotesize
The poles of $\zeta(s^2+h)$ are the vertices of dark hyperbola, indicated by two thick dots. The trivial zeroes of 
$\zeta(s^2+h)$ are indicated by thin dots on the imaginary axis; and the non-trivial zeroes are located on the darker 
painted region (critical region).}

\item [ii)]For $h<1$, $\Gamma$ is the interior of the dark hyperbola $Re(s)^2-Im(s)^2=1-h $ containing the imaginary axis:
\begin{center} 
\includegraphics[width=4cm, height=4cm]{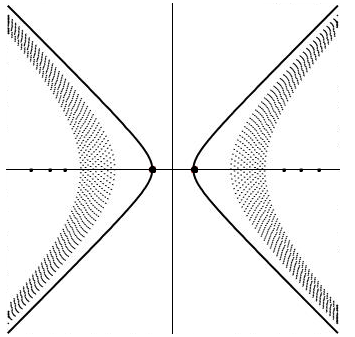}
\end{center}
{\footnotesize
The poles of $\zeta(s^2+h)$ are the vertices of dark hyperbola, indicated by two thick dots. The trivial zeroes of 
$\zeta(s^2+h)$ are indicated by thin dots on the real axis; the non-trivial zeroes are located on the darker painted 
region (critical region).}

\item [iii)]For $h=1$, $\Gamma$ is the interior of the cones
limited by the curves $y=|x|, y=-|x|$.
\begin{center} 
\includegraphics[width=4cm, height=4cm]{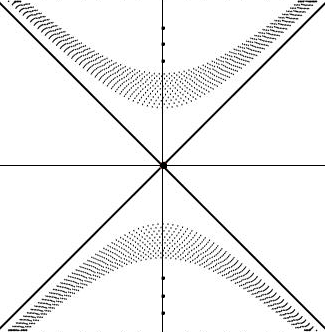}
\end{center}
{\footnotesize The pole of $\zeta(s^2+1)$ is the origin (vertex of
dark curves $y=|x|, y=-|x|$). The trivial zeroes of $\zeta(s^2+h)$
are indicated by thin dots on the imaginary axis; the non-trivial
zeroes are located on the darker painted region (critical
region).}
\end{itemize}


On the other hand, since the Riemann zeta function has an infinite number of zeroes on the critical strip 
(as famously proven by Hadamard and Hardy, see \cite{KaVo} for original references), we have that the function
$\zeta(s^2+h)$ also has an infinite number of zeroes; we denote the set of all such zeroes by $\mathcal{Z}$.
Using i), ii) and iii) we have that 
$$ \sup_{z\in \mathcal{Z}}|Re(z)|=+\infty.$$
This analysis implies that the expression
$\mathcal{L}^{-1}(\mathcal{L}(J)/\zeta_h)$ for the solution to equation (\ref{Zeq_044}) does not always make sense,
since the function $\mathcal{L}(J)/\zeta_h$ does not necessarily belongs to $H^p(\mathbb{C}_+)\,$. 

These observations mean that a new approach for the study of Equation (\ref{Zeq_04}) is necessary. 
We will present a method based on the Borel transform, see \cite{CPR,U}, in the fortcoming paper \cite{CPR_Borel}.

\paragraph{Acknowledgements}
A.C. has been supported by PRONABEC (Ministerio de Educaci\'on, Per\'u) and FONDECYT 
through grant \# 1161691; H.P. and E.G.R. have been
partially supported by the FONDECYT operating grants \# 1170571 and
\# 1161691 respectively. The authors also thank A. Koshelev and G. Calcagni for
their bibliographic remarks and their interest in our work.


\begin{thebibliography}{999999}
\renewcommand{\baselinestretch}{1.0}
\small
\bibitem{AV} I.Ya. Aref'eva, I.V. Volovich, Cosmological Daemon, {\em J. High Energy Physics} 2011 no. 08, Paper 102, 29 pp.
%
\bibitem{AV2} I.Ya. Aref'eva, I.V. Volovich, Quantization of the Riemann zeta-function and cosmology, {\em Intern. J. Geom. Meth. Modern Physics,} 4 (5), 881-895, 2007.
%
\bibitem{B}   N. Barnaby, A new formulation of the initial value
              problem for nonlocal
              theories, {\em Nuclear Physics B} 845 (2011), 1--29.
\bibitem{BBC} N. Barnaby, T. Biswas, J.M. Cline, $p$-adic inflation,
              {\em J. High Energy
              Physics} 2007, no. 04, Paper 056, 35 pp.
\bibitem{BK1} N. Barnaby, N. Kamran, Dynamics with infinitely many
              derivatives: the initial value problem, {\em J. High Energy
              Physics} 2008 no. 02, Paper 008, 40 pp.
\bibitem{BK2} N. Barnaby, N. Kamran, Dynamics with infinitely many
              derivatives: variable coefficient equations,
              {\em J. High Energy
              Physics} 2008 no. 12, Paper 022, 27 pp.
\bibitem{BG}  K. Bartkowski, P. G\'orka, One-dimensional Klein - Gordon
              equation with logarithmic nonlinearities,
              {\em J. Phys. A: Math. Theor.} 41 (2008), 355201.
\bibitem{BeHeJo}  J.J. Benedetto, H.P. Heinig, R. Johnson, Weighted Hardy spaces 
              and the Laplace transform II,
              {\em Math. Nachr.} 132 (1987), 29--55.
%
\bibitem{BCK} T. Biswas, J.A.R. Cambranos, J.I. Kapusta, Thermal
              duality and Hagedorn
              transition from $p$-adic strings, {\em Phys. Rev. Lett.}
              104 (2010), 021601.
\bibitem{BGKM} T. Biswas, E. Gerwick, T. Koivisto, A. Mazumdar,
               Towards singularity and ghost
               free theories of gravity, {\em Phys. Rev. Lett.}
               108 (2012), 031101.
\bibitem{Bou} C. Bourlet, Sur les op\'erations en g\'en\'eral et les \'equations 
         diff\'erentielles lin\'eaires d'ordre infini, {\em Ann. sci. de l'École 
         Normale Sup\'erieure} 14, 133-190, 1897.
\bibitem{Cal1} G. Calcagni, Cosmological tachyon from cubic string field theory,
              {\em J. High Energy Physics} 2006 no. 05, Paper 012, 26 pp.
\bibitem{CalMo} G. Calcagni and L. Modesto, Nonlocal quantum gravity and M-theory,
             {\em Physical Review D} 91 (2015), 124059, 16 pp.
\bibitem{CN} G. Calcagni and G. Nardelli, Cosmological rolling solutions of 
             nonlocal theories, {\em International Journal of Modern Physics D} 
             19 (2010), 329--338.
\bibitem{Car} R.D. Carmichael, Linear differential equations of
              infinite order,
              {\em Bull. Amer. Math. Soc.} 42 (1936), 193--218.
\bibitem{CPR} M. Carlsson, H. Prado, E.G. Reyes,
              Differential equations with infinitely many derivatives
              and the Borel transform,
              {\em Annales Henri Poincar\'e} 17 Issue 8 (2016),
              2049--2074.
\bibitem{CPR_Borel} A. Ch\'avez, H. Prado, E.G. Reyes,
        A Borel transform approach to linear
        zeta-nonlocal field equations, Preprint, 2017.
\bibitem{Chu} R. V. Churchill, The inversion of the Laplace
              transformation by a direct expansion in series and its
              application to boundary-value problems,
              {\em Math. Z.} 42 (1937), 567--579.
\bibitem{Doetsch} G. Doetsch, Introduction to the Theory and
                  Application of the Laplace Transformation,
                  Springer-Verlag, 1974.
\bibitem{Doetsch2} G. Doetsch, Bedingungen f{\"u}r die Darstellbarkeit einer 
                   Funktion als Laplace-integral und eine                                      
                   Umkehrformel f{\"u}r die Laplace-Transformation.(German), 
                   {\em Math. Z.} 42 (1937), no. 1, 263--286.
\bibitem{D} B. Dragovich, Zeta-nonlocal scalar fields,
           {\em Theoret. Math. Phys.}, 157 (2008), 1671--1677.
\bibitem{D1} B. Dragovich, Nonlocal dynamics of $p$-adic strings,
             {\em Theoret. Math. Phys.}, 164 (2010), 1151--1155.
\bibitem{D2} B. Dragovich, Towards effective lagrangians for adelic
             strings, {\em Fortschr. Phys.}, 57, No. 5--7 (2009), 546--551.
\bibitem{EW} D.A. Eliezer, R.P. Woodard, The problem of nonlocality
              in string theory,
             {\em Nuclear Physics B} 325 (1989), 389--469.
\bibitem{GS}  A. A. Gerasimov, S. L. Shatashvili,
              On exact tachyon potential in open string field theory,
              {\em J. High Energy Physics}  2000,  no. 10, Paper 34, 12 pp.
\bibitem{GPR_JMP} P. G\'orka, H. Prado, E.G. Reyes, Functional
                  calculus via Laplace
                  transform and equations with infinitely many
                  derivatives,
                  {\em J. Math. Phys.} 51 (2010), 103512.
\bibitem{GPR_CQG} P. G\'orka, H. Prado, E.G. Reyes, The initial value
                  problem for
                  ordinary differential equations with infinite many
                  derivatives,
                  {\em Class. Quantum Grav.}, 29 (2012) 065017 (15pp).
\bibitem{Har} Z. Harper, Laplace transform representations
                      and Paley--Wiener theorems for functions on
                      vertical strips, {\em Documenta Math.}
                      15 (2010) 235--254.
\bibitem{Hi} E. Hille, Analytic Function Theory, Vol. 1, Ginn and Company, 1959.
\bibitem{H} L. H\"ormander, The Analysis of Linear Partial
            Differential Operators III, Springer-Verlag, Berlin, 1985.
\bibitem{KaVo} A.A. Karatsuba, S.M. Voronin, The Riemann Zeta-Function,
               Translated from the Russian by Neal Koblitz.
               de Gruyter Expositions in Mathematics, 5.
               Walter de Gruyter \& Co., Berlin, 1992.
\bibitem{Ko}  A.S. Koshelev, Non-local SFT Tachyon and Cosmology,
              {\em J. High Energy Physics}  2007 no. 04, Paper 029, 17 pp.
\bibitem{M} J.A. Minahan, Mode interactions of the tachyon condensate
            in $p$-adic string theory,
            {\em J. High Energy Physics} 2001 no. 03, Paper 028, 17 pp.
\bibitem{Moeller} N. Moeller, B. Zwiebach, Dynamics with
                  infinitely many time derivatives and rolling
                  tachyons, {\em J. High Energy Physics}
                  2002, no. 10, Paper 34, 38 pp.
\bibitem{Mo1} J.W. Moffat, Ultraviolet complete electroweak model
              without a Higgs particle,
              {\em Euro. Phys. J. Plus} 126:53 (2011).
\bibitem{MN} D.J. Mulryne, N.J. Nunes, Diffusing nonlocal inflation:
             Solving the field equations as an initial value problem,
             {\em Physical Review D} 78, 063519 (2008), 16 pages.
\bibitem{Nelson} E. Nelson, Analytic vectors, {\em Ann. of Math.}
                 70 (1959), 572--615.
\bibitem{Ni} S. M. Nikol'skii, Approximation of Functions
             of Several Variables and Imbedding Theorems,
             Springer-Verlag, 1975.
\bibitem{PU} A. Pais, G.E. Uhlenbeck, On field theories with
             non-localized action,
             {\em Physical Review} 79, Number 1 (1950), 145--165.
%
\bibitem{Rooney} P. G. Rooney, Generalized $H_p$ spaces and Laplace Transforms,
            Abstract Spaces and Approximation (Proc. Conf.,
            Oberwolfach, 1968), pp. 258--269. Birkh¨auser, 1969.
%
\bibitem{Ta} W. Taylor, String Field Theory, In:
             Approaches to Quantum Gravity,
             Daniele Oriti (Ed.), 210--228, Cambridge University Press, 2009.
\bibitem{U} S. Umarov, Introduction to Fractional and
         Pseudo-Differential Equations with Singular Symbols, Springer, 2015.
\bibitem{V} V.S. Vladimirov, The equation of the $p$-adic open string
            for the scalar
            tachyon field, {\em Izvestiya: Mathematics} 69 (2005), 487--512.
\bibitem{VV} V.S. Vladimirov, Ya. I. Volovich, Nonlinear dynamics
             equation in
             p-adic string theory, {\em Teoret. Mat. Fiz.} 138 (2004), 355--368;
             English transl., {\em Theoret. Math. Phys.} 138 (2004), 297--309.
\bibitem{Vo} Y. Volovich, Numerical study of nonlinear with an infinite number
             of derivatives, {\em J. Phys. A: Math. Gen.} 36 (2003), 8685--8701.
\bibitem{W} E. Witten, Noncommutative geometry and string field
            theory, {\em Nuclear Physics B} 268 (1986), 253--294
\bibitem{KoYo} K. Yoshida, Functional Analysis, Springer, 1966.


\end{thebibliography}
\end{document}